\documentclass{elsarticle}

\usepackage{graphicx}
\usepackage{amssymb, amsmath}

\usepackage{color}
\usepackage{epsfig}
\usepackage{scalefnt}
\usepackage{verbatim}
\usepackage{psfrag}

\usepackage[ruled,vlined,english]{algorithm2e}
\usepackage{tikz}
\usepackage{float}
\usepackage{enumitem}
\usepackage{url}

\usepackage{thmtools}
\usepackage{thm-restate}

\declaretheorem[name=Theorem,numberwithin=section]{theorem}

\newtheorem{lemma}[theorem]{Lemma}
\declaretheorem[name=Claim,numberwithin=lemma]{claim}
\newtheorem{fact}{Fact}[claim]

\newdefinition{definition}[theorem]{Definition}
\newproof{proof}{Proof}

\let\doendproof\endproof
\renewcommand\endproof{~\hfill$\blacksquare$\doendproof}

\begin{document}
\begin{frontmatter}
\title{The maximum time of 2-neighbor bootstrap percolation: complexity results}
\tnotetext[t1]{This research was supported by CNPq (Universal Proc. 478744/2013-7) and FAPESP (Proc. 2013/03447-6). The statements of some of the results of this paper appeared in WG'2014}

\author[ufc]{Thiago Marcilon}
\ead{thiagomarcilon@lia.ufc.br}
\author[ufc]{Rudini Sampaio}
\ead{rudini@lia.ufc.br}

\address[ufc]{Departamento de Computa\c c\~ao, Universidade Federal do Cear\'a, Fortaleza, Cear\'a, Brazil}

\begin{abstract}
In $2$-neighborhood bootstrap percolation on a graph $G$, an infection spreads according to the following deterministic rule: infected vertices of $G$ remain infected forever and in consecutive rounds healthy vertices with at least $2$ already infected neighbors become infected. Percolation occurs if eventually every vertex is infected. The maximum time $t(G)$ is the maximum number of rounds needed to eventually infect the entire vertex set. In 2013, it was proved \cite{eurocomb13} that deciding whether $t(G)\geq k$ is polynomial time solvable for $k=2$, but is NP-Complete for $k=4$ and, if the problem is restricted to bipartite graphs, it is NP-Complete for $k=7$. In this paper, we solve the open questions. We obtain an $O(mn^5)$-time algorithm to decide whether $t(G)\geq 3$. For bipartite graphs, we obtain an $O(mn^3)$-time algorithm to decide whether $t(G)\geq 3$, an $O(m^2n^9)$-time algorithm to decide whether $t(G)\geq 4$ and we prove that $t(G)\geq 5$ is NP-Complete.
\end{abstract}


\end{frontmatter}

\section{Introduction}

We consider a problem in which an infection spreads over the vertices of a connected simple graph $G$ following a deterministic spreading rule in such a way that an infected vertex will remain infected forever. Given a set $S \subseteq V(G)$ of initially infected vertices, we build a sequence $S_{(0)}, S_{(1)}, S_{(2)}, \ldots$ in which $S_{(0)}=S$ and $S_{(i+1)}$ is obtained from $S_{(i)}$ using such spreading rule.

Under $r$-neighbor bootstrap percolation on a graph $G$, the spreading rule is a threshold rule in which $S_{(i+1)}$ is obtained from $S_{(i)}$ by adding to it the vertices of $G$ which have at least $r$ neighbors in $S_{(i)}$. Given a set $S$ of vertices and a vertex $v$ of $G$, let $t_r(G,S,v)$ be the minimum $t$ such that $v$ belongs to $S_{(t)}$ (let $t_r(G,S,v) = \infty$ if there is no such $t$). We say that a set $S_{(0)}$ infects $G$ (or that $S_{(0)}$ is a percolating set of $G$) if eventually every vertex of $G$ becomes infected, that is, there exists $t$ such that $S_{(t)} = V(G)$. If $S$ is a percolating set of $G$, then we define $t_r(G,S)$ as the minimum $t$ such that $S_{(t)} = V(G)$. Also, define the {\em percolation time of $G$} as $t_r(G) = \max \{t_r(G,S) : S \text{ infects } G\}$. In this paper, we shall focus on the case where $r=2$ and in such case we omit the subscript of the functions $t_r(G,S,v)$, $t_r(G,S)$ and $t_r(G)$.

Bootstrap percolation was introduced by Chalupa, Leath and Reich \cite{chalupa} as a model
for certain interacting particle systems in physics. Since then it has found applications
in clustering phenomena, sandpiles \cite{sandpiles}, and many other areas of statistical physics, as well as in neural networks \cite{neural2} and computer science \cite{cs1}.

There are two broad classes of questions one can ask about bootstrap percolation. The
first, and the most extensively studied, is what happens when the initial configuration $S_{(0)}$ is chosen randomly under some probability distribution? For example, vertices are included in $S_{(0)}$ independently with some fixed probability $p$. One would like to know how likely percolation is to occur, and if it does occur, how long it takes.

The answer to the first of these questions is now well understood for various graphs. An interesting case is the one of the lattice graph $[n]^d$, in which $d$ is fixed and $n$ tends to infinity, since the probability of percolation under the $r$-neighbor model displays a sharp threshold between no percolation with high probability and percolation with high probability. The existence of thresholds in the strong sense just described first appeared in papers by Holroyd, Balogh, Bollob\'as, Duminil-Copin and Morris \cite{holroyd,balogh,balogh2}. Sharp thresholds have also been proved for the hypercube (Balogh and Bollob\'as \cite{balogh3}, and Balogh, Bollob\'as and Morris \cite{balogh4}). There are also very recent results due to Bollob\'as, Holmgren, Smith and Uzzell \cite{bollobasHolmgrenSmithUzzell}, about the time percolation take on the discrete torus $\mathbb{T}^d_n = (\mathbb{Z}/n\mathbb{Z})^d$ for a randomly chosen set $S_{(0)}$.

The second broad class of questions is the one of extremal questions. For example, what is the smallest or largest size of a percolating set with a given property? The size of the smallest percolating set in the $d$-dimensional grid, $[n]^d$, was studied by Pete and a summary can be found in \cite{BaloghPete}. Morris \cite{morris} and Riedl \cite{riedl}, studied the maximum size of minimal percolating sets on the square grid and the hypercube $\{0,1\}^d$, respectively, answering a question posed by Bollob\'as. However, it was proved in \cite{ningchen,CDPRS2011} that finding the smallest percolating set is NP-complete for general graphs. Another type of question is: what is the minimum or maximum time that percolation can take, given that $S_{(0)}$ satisfies certain properties? Recently, Przykucki \cite{Przykucki} determined the precise value of the maximum percolation time on the hypercube $2^{[n]}$ as a function of $n$, and Benevides and Przykucki \cite{fabricio1,fabricio1b} have similar results for the square grid, $[n]^2$, also answering a question posed by Bollob\'as.

Here, we consider the decision version of the percolation time problem, as stated below.

\vspace{5 pt}
\noindent{\sc percolation time} \\
{\em Input:} A graph $G$ and an integer $k$. \\
{\em Question:} Is $t(G) \geq k$?
\vspace{5 pt}

In 2013, Benevides, Campos, Dourado, Sampaio and Silva \cite{eurocomb13}, among other results, proved that deciding whether $t(G)\geq k$ is polynomial time solvable when the parameter $k$ is fixed in $2$, but is NP-Complete when the parameter $k$ is fixed in $4$ and, when we restrict the problem to allow only bipartite graphs, is NP-Complete when the parameter $k$ fixed is fixed in $7$, thus, leaving open the questions: $t(G)\geq k$ is polynomial time solvable when the parameter $k$ is fixed in $3$? When we restrict the problem to bipartite graphs, what is the threshold, regarding the parameter $k$, between polynomiality and NP-completeness? In this paper, we solve these questions. 

In Section \ref{secaoBip5}, we show that the Percolation Time Problem is NP-Complete when restricted to bipartite graphs and $k\geq 5$. In Section \ref{secaopolyres}, we discuss the similarities between the structural characterizations presented in Sections \ref{secaoBip3}, \ref{secaoBip4} and \ref{secaoGer3} and the arguments used in the proofs in these sections. In Sections \ref{secaoBip3} and \ref{secaoBip4}, we show that the problem is polynomial time solvable for $k=3$ and $k=4$ by showing a structural characterization that can be computed in time $O(mn^3)$ and $O(m^2n^9)$ respectively. For general graphs, we show in Section \ref{secaoGer3} that the problem is polynomial time solvable for $k=3$ by showing a structural characterization that can be computed in time $O(mn^5)$. 

\subsection{Related works and some notation}

It is interesting to notice that infection problems appear in the literature under many different names and were studied by researches of various fields. The particular case in which $r=2$ in $r$-neighborhood bootstrap percolation is also a particular case of a infection problem related to convexities in graph.

A finite {\it convexity space} \cite{le,ve} 
is a pair $(V,\mathcal{C})$ consisting of a finite ground set $V$
and a set $\mathcal{C}$ of subsets of $V$ satisfying
$\emptyset, V \in \mathcal{C}$ and if $C_1,C_2\in \mathcal{C}$, then $C_1\cap C_2\in \mathcal{C}$.
The members of $\mathcal{C}$ are called {\em $\mathcal{C}$-convex sets}
and the {\em convex hull} of a set $S$ is the minimum convex set $H(S) \in \mathcal{C}$ containing $S$. 

A convexity space $(V,\mathcal{C})$ is an {\it interval convexity} \cite{cal}
if there is a so-called {\it interval function} \mbox{$I:\binom{V}{2}\to 2^V$}
such that a subset $C$ of $V$ belongs to $\mathcal{C}$
if and only if $I(\{ x,y\})\subseteq C$ for every two distinct elements $x$ and $y$ of $C$. With no risk of confusion, for any $S \subseteq V$, we also denote by $I(S)$ the union of $S$ with $\bigcup_{x,y\in S} I(\{x,y\})$. In interval convexities, the convex hull of a set $S$ can be computed by exhaustively applying the corresponding interval function until obtaining a convex set.


The most studied graph convexities defined by interval functions are those in which $I(\{x,y\})$ is the union of paths between $x$ and $y$ with some particular property. Some common examples are the $P_3$-convexity~\cite{er1972}, geodetic convexity~\cite{faja} and monophonic convexity~\cite{du1988}. We observe that the spreading rule in $2$-neighbors bootstrap percolation is equivalent to $S_{(i+1)} = I(S_{(i)})$ where $I$ is the interval function which defines the $P_3$-convexity: $I(S)$ contains $S$ and every vertex belonging to some path of 3 vertices whose extreme vertices are in $S$. For these reasons, we call a percolating set by \emph{hull set}.

In geodetic convexity, where the interval of $S$ contains $S$ and every vertex lying in some geodesic joining two vertices of $S$, it was defined the {\em geodetic iteration number of a graph}~\cite{cps2002,hn1981} which is analogous to our definition of percolation time.
Related to the geodetic convexity, there exists the {\em geodetic iteration number of a graph}~\cite{cps2002,hn1981}, which is similar to the percolation time.

\section{$t(G)\geq 5$ is NP-Complete in bipartite graphs}
\label{secaoBip5}

In \cite{eurocomb13}, it was proved that deciding if $t(G)\geq 7$ is NP-Complete in bipartite graphs. The following theorem improves this result.

\begin{theorem}
Deciding whether $t(G)\geq k$ is NP-Complete in bipartite graphs for any $k\geq 5$.
\end{theorem}

\begin{proof}

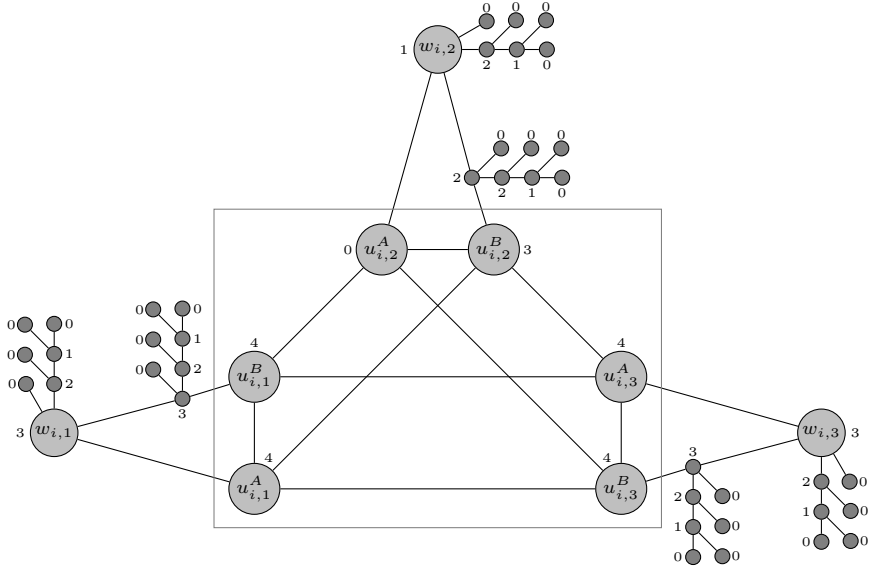
\begin{figure}[ht]
\begin{tikzpicture}[scale=.85]
\tikzstyle{every node}=[font=\scriptsize]
\tikzstyle{vertex}=[draw,circle,fill=black!25,minimum size=14pt,inner sep=1pt]
\tikzstyle{leafvertex}=[draw,circle,fill=black!50,minimum size=2pt,inner sep=2pt]
\def \diff {17}
\def \dist {3cm}
\foreach \i/\im [evaluate=\i as \angle using (\i*90)-90] in {1/3,2/2,3/1}{
	\node[vertex] (uB\i) at (\angle - \diff:\dist) {$u_{i,\im}^B$};
    \node[vertex] (uA\i) at (\angle + \diff:\dist) {$u_{i,\im}^A$};
    \node[vertex] (w\i)  at (\angle:\dist+3cm) {$w_{i,\im}$};
    \node[leafvertex,shift={(90*\i-90 + \diff:1cm)}] (l\i0) at (uB\i) {};
    \draw[-] (uA\i) to (uB\i);
    \draw[-] (uA\i) to (w\i);
    \draw[-] (uB\i) to (l\i0);
    \draw[-] (l\i0) to (w\i);
    
    \node[leafvertex,shift={(90*\i-180:0.65cm)}] (lw\i1) at (w\i) {};
    \node[shift={(90+\i*90:0.2cm)}] (nlw\i1) at (lw\i1) {\tiny{2}};
    \node[leafvertex,shift={(90*\i-150:0.75cm)}] (lwb\i0) at (w\i) {};
    \node[shift={(270+\i*90:0.18cm)}] (nlwb\i0) at (lwb\i0) {\tiny{0}};
    
    \draw[-] (w\i) to (lw\i1);
    \draw[-] (w\i) to (lwb\i0);
    
    \foreach \j/\jm/\t in {2/1/1,3/2/0}{
    	\node[leafvertex,shift={(90*\i-180:0.25cm +\j*0.4cm)}] (lw\i\j) at (w\i) {};
        \node[shift={(90+\i*90:0.2cm)}] (nlw\i\j) at (lw\i\j) {\tiny{\t}};
        \draw[-] (lw\i\j) to (lw\i\jm);
    }
    \foreach \j in {1,2}{
    	\node[leafvertex,shift={(90*\i-135:0.55cm)}] (lwb\i\j) at (lw\i\j) {};
        \node[shift={(270+\i*90:0.18cm)}] (nlwb\i\j) at (lwb\i\j) {\tiny{0}};
        \draw[-] (lwb\i\j) to (lw\i\j);
    }
    \foreach \j/\jm/\t in {1/0/2,2/1/1,3/2/0}{
    	\node[leafvertex,shift={(90*\i-180:\j*0.4cm)}] (l\i\j) at (l\i0) {};
        \node[shift={(90+\i*90:0.2cm)}] (nl\i\j) at (l\i\j) {\tiny{\t}};
        \draw[-] (l\i\j) to (l\i\jm);
    }
    \foreach \j in {0,1,2}{
    	\node[leafvertex,shift={(90*\i-135:0.55cm)}] (lb\i\j) at (l\i\j) {};
        \node[shift={(270+\i*90:0.18cm)}] (nlb\i\j) at (lb\i\j) {\tiny{0}};
        \draw[-] (l\i\j) to (lb\i\j);
    }
}
\draw[-] (uA1) to (uB2);
\draw[-] (uA1) to (uB3);
\draw[-] (uA2) to (uB1);
\draw[-] (uA2) to (uB3);
\draw[-] (uA3) to (uB1);
\draw[-] (uA3) to (uB2);
\draw[gray] (135:4.95cm) rectangle (-\diff-6:3.8cm);
\node[shift={(0:0.45cm)}] (nw1) at (w1) {\tiny{3}};
\node[shift={(180:0.45cm)}] (nw2) at (w2) {\tiny{1}};
\node[shift={(180:0.45cm)}] (nw3) at (w3) {\tiny{3}};
\node[shift={(90:0.2cm)}] (nl10) at (l10) {\tiny{3}};
\node[shift={(180:0.2cm)}] (nl20) at (l20) {\tiny{2}};
\node[shift={(270:0.2cm)}] (nl30) at (l30) {\tiny{3}};
\node[shift={(90:0.45cm)}] (uA1) at (uA1) {\tiny{4}};
\node[shift={(115:0.45cm)}] (nuB1) at (uB1) {\tiny{4}};
\node[shift={(65:0.45cm)}] (uA3) at (uA3) {\tiny{4}};
\node[shift={(90:0.45cm)}] (nuB3) at (uB3) {\tiny{4}};
\node[shift={(180:0.45cm)}] (uA2) at (uA2) {\tiny{0}};
\node[shift={(0:0.45cm)}] (nuB2) at (uB2) {\tiny{3}};
\end{tikzpicture}
\caption{\label{gadget-bip-t5} Bipartite gadget for each clause $C_i$}
\end{figure}

Let us prove that the Percolation Time Problem is NP-Complete by showing a polynomial reduction from the problem \textbf{3-SAT}.
 Given $m$ clauses $\mathcal{C}=\{C_1,\ldots,C_m\}$ on variables $X=\{x_1, \ldots,x_n\}$ of an instance of \textbf{3-SAT}, let us denote the three literals of $C_i$ by $\ell_{i,1}$, $\ell_{i,2}$ and $\ell_{i,3}$. The reduction is as follows:\\
For each clause $C_i$ of $\mathcal{C}$, add to $G$ a gadget like the one in Figure \ref{gadget-bip-t5}. Then, for each pair of literals $\ell_{i,a}, \ell_{j,b}$ such that one is the negation of the other, add a vertex $y_{(i,a),(j,b)}$ adjacent to $w_{i,a}$ and $w_{j,b}$. Let $Y$ be the set of all vertices created this way. Finally, add a vertex $z$ adjacent to all vertices in $Y$ and a vertex $z'$ adjacent only to $z$. Denote the sets $\{u^A_{i,1}, u^A_{i,2}, u^A_{i,3}\},\{u^B_{i,1}, u^B_{i,2}, u^B_{i,3}\}$ and $\{w_{i,1}, w_{i,2}, w_{i,3}\}$ by $U^A_i,U^B_i$  and $W_i$, respectively. Let $U_i = U^A_i \cup U^B_i$, $U = \bigcup_{1\le i\le m} U_i$ and $W = \bigcup_{1\le i\le m} W_i$. Also, let $T$ be the smaller and darker vertices in the Figure \ref{gadget-bip-t5} in all the gadgets.

We have that each gadget of $G$ is bipartite. Consider a bipartition where all vertices in $U^B_i \cup W$, for all $1 \leq i \leq m$, are in the same partition, and all vertices in $U^A_i$, for all $1 \leq i \leq m$, are on the other partition. If you consider $Y \cup \{z'\}$ to be in the same partition as the vertices in $U^A_i$ and $z$ to be in the same partition as the vertices in $U^B_i \cup W$, we can always choose in which partition to put the vertices in $T$ so we can have a bipartition for $G$. Thus, $G$ is bipartite.

Now, let us show that $\mathcal{C}$ is satisfiable if and only if $G$ contains a hull set that infects all vertices in $G$ in time at least $k$, for any $k \geq 5$. First, we will consider the case $k=5$.

Suppose that $\mathcal{C}$ has a truth assignment. For each clause $C_i$, let $K_i$ denote the subset of $\{1,2,3\}$ such that $\ell_{i,k}$ is true for all $k \in K_i$. Let $S = \{z'\} \cup \{u^A_{i,k} : k \in K_i, 1 \leq i \leq m\} \cup T$. We can see in Figure \ref{gadget-bip-t5} that all vertices in the clause gadgets are infected in time at most 4. Also, we have that $S$ infects $w_{i,k}$, for any $k \in K_i$, in time 1 and infects $w_{i,k'}$, for any $k' \notin K_i$, in time 3 for every $1 \leq i \leq m$. Since we used a truth assignment, we have that all vertices of $Y$ are infected in time exactly 4 because any vertex $y \in Y$ is adjacent to exactly one vertex in $W$ infected at time 1 and another vertex in $W$ infected at time 3. Thus, the vertex $z$ is infected in time 5 and, therefore, $t(G) \geq 5$.

Now, suppose that there is a instance $\mathcal{C}$ is not satisfiable. Let $G$ be the graph resulting from the reduction applied to $\mathcal{C}$. Let us prove that $t(G) < 5$. Let $S$ be any hull set of $G$. Since $S$ is a hull set, then it must have at least one vertex in each set $U_i$, because each vertex in $U_i$ have only one neighbor outside $U_i$, and all vertices of degree one, for the same reason. Thus, since $\mathcal{C}$ is not satisfiable, then there is a vertex $y \in Y$ such that $S$ infects $y$ at time $\leq 3$. This is because, if $S$ infects all vertices in $Y$ at time 4, then we can find an assignment to the variables of $\mathcal{C}$ that satisfies $\mathcal{C}$. Consider the following assignment: If a vertex $u_{i,j}$ is in $S$ and $\ell_{i,j}$ is a positive literal then assign true to the variable that it represents and, hence, if a vertex $u_{i,j}$ is in $S$ and $\ell_{i,j}$ is a negative literal then assign false to the variable that it represents. After that, if there is a variable that does not yet have a value, assign true to that variable. This assignment satisfies $\mathcal{C}$ because since $S$ has at least one vertex in each set $U_i$, then each clause will have at least one literal that evaluates to true. Furthermore, since all vertices in $Y$ are infected at time 4, then we have that two vertices in $U$ that represent the same variable but negate each other, cannot be in $S$ simultaneously and, therefore, we do not risk to assign different values to the same variable.

Thus, we have that $z$ is infected by $S$ at time $\leq 4$. Since any hull set $S$ infects all vertices in $Y \cup U$ at time $\leq 4$, then $x$ is the only vertex that can possible be infected at time greater than 4. Since, for any hull set $S$, $S$ infects $x$ at time $\leq 4$, hence all hull sets infect $G$ at time $\leq 4$. Therefore, $t(G) \leq 4$.

For values $k>5$, it suffices to alter the reduction by adding a path $P$ of length $k-5$ and linking one end of $P$ to the vertex $z$, appending a new leaf vertex to each vertex in $P$. The proof remains the same.
\end{proof}

\section{Discussion about the polynomial results and some useful lemmas}
\label{secaopolyres}

In the next sections, we prove that the Percolation Time Problem for $k=3$ in general graphs and for $k=3$ and $k=4$ for bipartite graphs is polynomial time solvable. In each of these cases, we obtain algorithms that are by-product of a structural characterization that can be computed in polynomial time. All the structural characterizations are very similar to each other: a graph $G$ has percolation time at least $k$ if and only if there is a set $F$ of vertices with size $|F| = O(1)$, a vertex $u\in V(G)$ and a set $T_0$ of vertices satisfying some specific restrictions, such that the set $F \cup T_0 \cup N_{\geq k}(u)$ infects some vertex $x$ at time $k$. Since $|F| = O(1)$, the difficulty to compute these characterizations resides in how many sets $T_0$ we have to test. In each of the next sections, we define the set $T_0$ and, when it is not obvious, we prove that the number of sets $T_0$ that we have to test and the construction of a set $T_0$ can be done in polynomial time.

Also, in the proofs of the next sections, we will often prove that some vertex $v$ is infected by some set $S$ at time $k$. To do that, we will prove that $t(G,S,v)\geq k$ and $t(G,S,v)\leq k$. Recall the infection process $S_{(0)},S_{(1)},S_{(2)},\ldots$ beginning with $S_{(0)}=S$. Besides Lemmas \ref{lemadistanciak} and \ref{corolariosubset}, there are essentially two arguments that we will use to prove these two points. The first argument is that, if $S_{(k-1)}$ contains at least two vertices in the neighborhood of $v$, then $t(G,S,v)\leq k$ (that is, $v\in S_{(k)}$). The second argument is that, if $S_{(k-1)}$ contains at most one vertex in the neighborhood of $v$, then either $v$ is not infected by $S$ or $v$ is infected by $S$ at time at least $k$, i.e., $t(G,S,v)\geq k$.

It is worth noting that $t(G,S,v)\geq k$ means that either $v$ is infected by $S$ at time at least $k$ or $v$ is not infected by $S$. However, if $S$ is a hull set, then it means that $v$ is infected by $S$ at time at least $k$. Moreover, if we already know that a subset of $S$ infects $v$, then $t(G,S,v)\geq k$ means that $S$ infects $v$ at time at least $k$.

The following useful lemma show us that we can obtain a hull set for any time smaller or equal to $t(G)$.
\begin{lemma}
\label{lemainfectaexato}
Let $G$ be a graph. Then, for any $0\leq k\leq t(G)$, there exists a hull set $S$ of $G$ such that $t(G,S)=k$.
\end{lemma}

\begin{proof}
Given a hull set $S'$ of $G$ such that $t(G,S')\geq 1$, let $S=S'\cup S'_1$, where $S'_1$ is the set of vertices infected by $S'$ at time 1. Clearly, $t(G,S)=t(G,S')-1$. Applying this fact, we are done.
\end{proof}

We say that a set $T$ of vertices is \emph{co-convex} if every vertex of $T$ has at most one neighbor outside $T$. The next lemma proves that every hull set must contain at least one vertex of any co-convex set.
\begin{lemma}
\label{lemacoconvexo}
Let $S$ be a hull set of $G$ and let $T$ be a co-convex set of $G$. Then $S$ contains a vertex of $T$.
Consequently, $S$ contains all vertices of $G$ with degree 1.
Moreover, for all pair of adjacent vertices $x,y$ with degree two, $S$ contains at least one of them.
\end{lemma}

\begin{proof}
Suppose, by contradiction, that $S\cap T=\varnothing$. Since $S$ is a hull set, then $S$ infects all the vertices in $T$. Let $v$ be the first vertex in $T$ infected by $S$, i.e., $t(G,S,v) = \min_{v' \in T} t(G,S,v')$. Let $t(G,S,v) = t$. Since $v$ is infected at time $t>0$ by $S$, then $v$ must have two neighbor vertices $v_1$ and $v_2$ infected at time $t(G,S,v_1)=t-1$ and $t(G,S,v_2)\leq t-1$. Clearly $t(G,S,v_1)<t(G,S,v)$ and $t(G,S,v_2)<t(G,S,v)$. Since $v$ has at most one neighbor vertex outside $T$, then at least one of the two vertices $v_1$ and $v_2$ is in $T$, which is a contradiction because there is no vertex in $T$ infected earlier than $v$. Therefore, there must be a vertex in $T \cap S$.
If $z$ is a vertex of degree 1, then $T=\{z\}$ is clearly co-convex.
If $x,y$ are adjacent vertices with degree two, then $T=\{x,y\}$ is clearly co-convex.
\end{proof}

Given a vertex $v$, let $N(v)$ be the set of neighbors of $v$ and let $N[v]=N(v)\cup\{v\}$.
Given an integer $i\geq 0$, $N_i(v)$ denotes the set of vertices at distance $i$ of $v$; $N_{\geq i}(u)$ denotes the set of vertices at distance greater or equal to $i$ of $v$; and $N_{\leq i}(u)$ denotes the set of vertices at distance less or equal to $i$ of $v$.
The following lemma show us that, if a vertex $w$ is infected at time at least $k$ by a set $S$, then this also happens for $S\cup\{z\}$ for any $z\in N_{\geq k}(w)$. That is, $z$ cannot help a faster infection of $w$ because $z$ is too far from $w$.

\begin{lemma}
\label{lemadistanciak}
Let $S$ be a set of vertices of a graph $G$ and $w$ be a vertex of $G$. If $t(G,S,w)\geq k$, then $t(G,S,w)\geq t(G,S\cup\{z\},w)\geq k$ for any $z\in N_{\geq k}(w)$. Consequently, $t(G,S,w)\geq t(G,S\cup N_{\geq k}(w),w)\geq k$.
\end{lemma}

\begin{proof}
Let $S'=S\cup\{z\}$. Clearly, $t(G,S',w)\leq t(G,S,w)$.
Let us prove that $t(G,S',w)\geq k$ by induction on $k$. The proof for $k=0$ is trivial.
Now, let $k>0$ and suppose that this lemma holds for any value smaller than $k$. Let us prove that it also holds for $k$. Let $z\in N_{\geq k}(w)$ and suppose that $t(G,S,w)\geq k$. Thus, if $u\in N(w)$, then $z\in N_{\geq k-1}(u)$. By the inductive hypothesis, for every $u\in N(w)$, if $t(G,S,u)\geq k-1$, then $t(G,S',u)\geq k-1$. Therefore, since we have at most one vertex $x\in N(w)$ such that $t(G,S,x)<k-1$ (otherwise, $t(G,S,w)<k$), there is at most one vertex $x$ in $N(w)$ such that $t(G,S',x)<k-1$. Thus, $t(G,S',w)\geq k$ and we are done.
\end{proof}

\begin{lemma}
\label{corolariosubset}
Let G be a graph and let $S\subseteq S'\subseteq V(G)$. Then $t(G,S,u)\geq t(G,S',u)$ for any vertex $u$.
\end{lemma}

\begin{proof}
If $t(G,S,u)=\infty$, we are done. Suppose that $t(G,S,u)=k$. For any $0\leq i\leq k$, let $S_{(i)}$ and $S'_{(i)}$ be the sets of vertices infected by $S$ and $S'$ at time $i$. Notice that, if $S_{(i)}\subseteq S'_{(i)}$, then $S_{(i+1)}\subseteq S'_{(i+1)}$. Since $S=S_{(0)}\subseteq S'_{(0)}=S'$, then we have by induction that $S_{(k)}\subseteq S'_{(k)}$ and we are done.
\end{proof}

\section{Maximum Percolation Time with parameter fixed $k=3$ in bipartite graphs}
\label{secaoBip3}

The following theorem is the main result of this section.
\begin{restatable}{theorem}{biptt}
Deciding whether $t(G)\geq 3$ is $O(mn^3)$-time solvable in bipartite graphs.
\end{restatable}

To prove this, we first show an important structural result.

\begin{lemma}\label{lema-bip3}
Let $G$ be a connected bipartite graph and $T_0$ be the set of vertices of $G$ that have degree 1. $t(G)\geq 3$ if and only if there are three vertices $u$, $v \in N(u)$ and $s \in N_2(u)$ such that $T_0\cup N_{\geq 3}(u)\cup \{v,s\}$ infects $u$ at time 3.
\end{lemma}

\begin{proof}
Firstly, suppose that $t(G)\geq 3$. From Lemma \ref{lemainfectaexato}, there exists a hull set $S''$ that infects $G$ at time 3 and a vertex $u$ that is infected by $S''$ at time 3. From Lemma \ref{lemacoconvexo}, $T_0 \subseteq S''$. From Lemma \ref{lemadistanciak}, $S'=S''\cup N_{\geq 3}(u)$ is also a hull set that infects $u$ at time 3. If $S'$ contains a vertex in $N(u)$, let $S=S'$. Otherwise, let $v$ be a neighbor of $u$ with the smallest infection time with respect to the hull set $S'$ and let $S=S'\cup\{v\}$. Since $G$ is bipartite, the distance from $v$ to any other vertex of $N(u)$ is at least two. Since all vertices in $N(u)-\{v\}$ are infected at time $\geq 2$ by $S'$, then by Lemma \ref{lemadistanciak} all vertices in $N(u)-\{v\}$ are infected at time $\geq 2$ by $S$. Thus, $S$ infects $u$ at time $\geq 3$. In fact, $S$ infects $u$ at time exactly 3, because $S'\subseteq S$ and by Lemma \ref{corolariosubset}. Note that, since $v$ is the only vertex of $N(u)$ infected by $S$ at time 0 and $G$ is bipartite, then any vertex in $N_2(u)$ infected at time 1 by $S$ is also infected at time 1 by the set $\{v\} \cup N_{\geq 3}(u) \cup T_0$. Now, we split the proof in two cases.

The first case occurs when no vertex of $N(u)-\{v\}$ has a neighbor in $S$. Let $w\in N(u)-\{v\}$ be a vertex such that $t(G,S,w)=2$ and let $s,z$ be vertices of $N(w)$ infected by $S$ at time 1. Since $G$ is bipartite, $s,z\in N_2(u)$. Since no vertex of $N(u)-\{v\}$ has a neighbor in $S$, then any vertex of $N(u)-\{v\}$ has at most one neighbor in $S \cup\{s\}$. Therefore, every vertex of $N(u)-\{v\}$ is infected by $S\cup\{s\}$ at time $\geq 2$. Thus, $S\cup\{s\}$ infects $u$ at time $\geq 3$. In fact, by Lemma \ref{corolariosubset}, $S\cup\{s\}$ infects $u$ at time exactly 3.
Since $t(G,S,z) = 1$, then the set $\{v,s\}\cup N_{\geq 3}(u)\cup T_0$ infects $z$ at time 1. Moreover, $\{s,v\}\cup N_{\geq 3}(u)\cup T_0\subseteq S\cup\{s\}$ and, by Lemma \ref{corolariosubset}, the set $\{s,v\}\cup N_{\geq 3}(u)\cup T_0$ infects $w$ at time exactly 2. With the same arguments, we have that $\{s,v\}\cup N_{\geq 3}(u)\cup T_0$ infects $u$ at time 3, through the vertices $w$ and $v$.

The second case occurs when there is a vertex $w\in N(u)-\{v\}$ with a neighbor $s\in S$. Let $z$ be any vertex of $N(w)$ infected by $S$ at time 1. Again, we have that $s,z \in N_2(u)$. Similarly to the previous case, since the set $\{v,s\} \cup N_{\geq 3}(u) \cup T_0$ infects $z$ at time 1 and $\{v,s\}\cup N_{\geq 3}(u)\cup T_0\subseteq S$, then by Lemma \ref{corolariosubset} the set $\{s,v\}\cup N_{\geq 3}(u)\cup T_0$ infects $w$ at time 2 and $u$ at time 3.
With this, we conclude the proof of the necessary condition.

Now, let us prove the sufficient condition. Suppose that there are three vertices $u$, $v$ and $s$ such that $v\in N(u)$, $s\in N_2(u)$ and the set $S_0=T_0\cup N_{\geq 3}(u)\cup \{v,s\}$ infects $u$ at time 3. 
We will construct a hull set $S\supseteq S_0$ in rounds. Let $S_i$ be the set obtained in round $i$. Consider we are in round $i\geq 0$. If $S_i$ is a hull set, let $S=S_i$ and we stop. Otherwise, suppose that there is a vertex $q\in N_2(u)$ not infected by $S_i$. Notice that $q$ is adjacent to a vertex $w\in N(u)$ not infected by $S_i$, because $q$ has degree at least 2 and $G$ is bipartite. Let $S_{i+1}=S_i\cup\{q\}$. Since $q$ is at distance 2 of $u$, then, by Lemma \ref{lemadistanciak}, $2\leq t(G,S_{i+1},u)\leq 3$. Therefore, $S_{i+1}$ infects $w$ at time $\geq 3$, since all neighbors of $w$ are at distance $\geq 2$ of $q$.

Since $N_2(u)$ is finite, this procedure will finish and all vertices of $N_2(u)$ are infected by $S$ and consequently $S$ is a hull set. Moreover, $2\leq t(G,S,u)\leq 3$ and $S$ infects at time $\geq 3$ some vertex $w\in N(u)$.
\end{proof}

Given the vertices $u,v,s$ such that $v\in N(u)$, $s\in N_2(u)$ and a set $S_0=T_0\cup N_{\geq 3}(u)\cup \{v,s\}$ that infects $u$ at time 3, the construction to build a hull set $S$ such that $t(G,S) \geq 3$ that is described in Lemma \ref{lema-bip3} can be implemented to run in time $O(mn^2)$. This is because, since $G$ is connected, the sets $N(u),N_2(u)$ and $N_{\geq 3}(u)$ can be computed in time $O(m)$ and, in each step $i$, we only have to find one vertex, if there is any, belonging to the set $N_2(u) \cap Y$, which can be done in time $O(n)$, and then update the set $Y$, which can be done in time $O(mn)$. Now, let us prove the main theorem of the section.

\biptt*

\begin{proof}
Consider the following algorithm:\\
\begin{algorithm}[ht]
\label{algbip3}
\DontPrintSemicolon
\KwIn{A bipartite connected graph $G$.}
\KwOut{\textbf{Yes}, if $t(G) \geq 3$. \textbf{No}, otherwise.}
\ForAll{$u \in V(G)$}{
	\ForAll{$v \in N(u)$}{
    	\ForAll{$s \in N_2(u)$}{
    		\If{$N_{\geq 3}(u) \cup T_0 \cup \{s,v\}$ infects $u$ at time 3}{
				\Return {\textbf{Yes}}
			}
        }
    }
}
\Return {\textbf{No}}

\caption{Algorithm that solves the Maximum Percolation Problem for $k = 3$ in bipartite graphs}
\end{algorithm}

The Algorithm \ref{algbip3} outputs \textbf{Yes} if and only if, given a bipartite graph $G$, there are vertices $u$, $v \in N(u)$ and $s$ in $N_2(u)$ such that the set $\{s,v\} \cup N_{\geq 3}(u) \cup T_0$ infects $u$ at time 3. Thus, by Lemma \ref{lema-bip3}, the algorithm outputs \textbf{Yes} if and only if $t(G) \geq 3$. Furthermore, the Algorithm \ref{algbip3} is an $O(mn^3)$-time algorithm. This is because in $O(n)$-time the set $T_0$ can be computed, given a vertex $u$, it is necessary $O(m)$-time to compute the sets $N(u),N_2(u)$ and $N_{\geq 3}(u)$, and given $u,v$ and $s$, it is necessary $O(m)$-time to test whether the set $\{s,v\} \cup N_{\geq 3}(u) \cup T_0$ infects $u$ at time 3. Therefore, the algorithm \ref{algbip3} decides whether $t(G) \geq 3$, for any bipartite connected graph $G$, in $O(mn^3)$-time.
\end{proof}

\section{Maximum Percolation Time with parameter fixed $k=3$}
\label{secaoGer3}

The following theorem is the main result of this section.

\begin{restatable}{theorem}{gertt}
Deciding whether $t(G)\geq 3$ is $O(mn^5)$-time solvable.
\end{restatable}

In order to prove it, we first define the family of sets $\mathcal{T}_0^u$ and prove a structural result.

\begin{definition}
\label{definicaot0u}
Let $\mathcal{T}_0^u$ be the family of subsets of $V(G)$ such that a set of vertices $T_0\in\mathcal{T}_0^u$ if and only if, for every separator $v$ and every connected component $H_{v,i}$ of $G-v$ such that $u\notin V(H_{v,i})$ and $V(H_{v,i}) \subseteq N(v)$, $T_0$ contains exactly one vertex of $H_{v,i}$, and every vertex of $T_0$ satisfies this property.
\end{definition}

\begin{lemma}
\label{corolariot0u}
If $S$ is a hull set, then, for every vertex $u$, there is a set $T_0\in\mathcal{T}_0^u$ such that $T_0\subseteq S$.
\end{lemma}

\begin{proof}
Notice that, if $v$ is a separator and $H_{v,i}$ is a connected component of $G-v$ such that $u\not\in V(H_{v,i})$ and $V(H_{v,i})\subseteq N(v)$, then $V(H_{v,i})$ is co-convex.
Then, by Lemma \ref{lemacoconvexo}, $S$ contains one vertex of $H_{v,i}$.
Let $T_0$ be the set consisting of all such vertices. Thus, we have that $T_0\subseteq S$ and, by Definition \ref{definicaot0u}, that $T_0\in\mathcal{T}_0^u$.
\end{proof}

\begin{lemma}\label{lema-geral3}
Let $G$ be a connected graph. $t(G)\geq 3$ if and only if there is a vertex $u$, a set $T_0\in\mathcal{T}^u_0$ and a set $F$ with at most 4 vertices such that $T_0\cup N_{\geq 3}(u)\cup F$ infects $u$ at time 3.
\end{lemma}

\begin{proof}
First, let us prove the necessary condition. Suppose that $t(G)\geq 3$. Then there exists a hull set $S'$ of $G$ and a vertex $u$ such that $t(G,S',u)=3$. By Lemma \ref{lemadistanciak}, $S=S'\cup N_{\geq 3}(u)$ is a hull set such that $t(G,S,u)=3$. Since $S$ infects $u$ at time 3 and, to be infected, each vertex needs two infected neighbors, then there is a set $F\subseteq S$ with size $|F|\leq 2^3 = 8$ which infects $u$ at time 3. Thus, since, by Lemma \ref{corolariot0u}, there is a set $T_0 \in \mathcal{T}^u_0$ such that $T_0 \subseteq S$, then $F\subseteq T_0\cup N_{\geq 3}(u)\cup F \subseteq S$. By Lemma \ref{corolariosubset}, we have that $3 = t(G,S,u) \leq t(G,T_0 \cup N_{\geq 3}(u) \cup F,u) \leq t(G,F,u) = 3$. Therefore, $T_0\cup N_{\geq 3}(u)\cup F$ infects $u$ at time 3. Furthermore, we can decrease the size of $F$ to 4 by the following claim, proved in Section \ref{apendClaim}.

\begin{restatable}{claim}{claimmf}
\label{claimmf}
Let $G$ be a graph and $u\in V(G)$.
If $G$ has a hull set which infects $u$ at time 3, then $G$ has a set $F$ with at most 4 vertices and a hull set $S$ such that $F\subseteq S$ and $t(G,F,u)=t(G,S,u)=3$.
\end{restatable}

With this, we conclude the proof of the necessary condition.

Now, let us prove the sufficient condition.
Suppose that there is a vertex $u$, a set $T_0\in\mathcal{T}^u_0$ and a set $F \subseteq V(G)$ with $|F|\leq 4$ such that $S_0=T_0\cup N_{\geq 3}(u)\cup F$ infects $u$ at time 3. 
If $S_0$ is a hull set, we are done. Thus, assume that $S_0$ is not a hull set.
We will show how to construct a hull set $S$ such that $t(G,S)\geq 3$. We begin with $S=S_0$. Let $S_i$ be the constructed set at step $i$. Each step adds one vertex to $S$ and, at the end of each step, it is guaranteed that $S_i$ infects some vertex $u_i$ at time $\geq 3$ ($u_0=u$) and infects $u$ at time $\geq 2$. Let $X_i$ be the set of vertices infected by $S_i$, and let $Y_i = V(G)-X_i$. It is worth noting that a vertex in $Y_i$ is adjacent to at most one vertex in $X_i$ and, since it must have degree at least two, it is also adjacent to at least one vertex in $Y_i$.

For $i\geq 0$, in the $(i+1)$th step of the construction, assume that there exists a vertex $y_i\in Y_i\cap N_2(u_i)$ with no neighbor infected by $S_i$ at time $\geq 2$. Let $S'_{i+1}=S_i\cup\{y_i\}$. Since $S_i$ infects $u_i$ at time $\geq 3$, $u_i$ has at most one neighbor infected by $S_i$ at time $\leq 1$ and, since $y_i\in N_2(u)$, $u_i$ is not adjacent to $y_i$. Furthermore, we have that:
\begin{itemize}
\item Every neighbor of $u_i$ infected by $S_i$ at time $\geq 2$ is also infected by $S'_{i+1}$ at time $\geq 2$. This is because $y_i$ is at distance $\geq 2$ of every vertex infected by $S_i$ at time $\geq 2$, and, therefore, by Lemma \ref{lemadistanciak}, every vertex infected by $S_i$ at time $\geq 2$, including the neighbors of $u_i$, is infected by $S'_{i+1}$ at time $\geq 2$.
\item Every neighbor of $u_i$ not infected by $S_i$ is either not infected by $S'_{i+1}$ or infected by $S'_{i+1}$ at time $\geq 2$. This is because to a vertex $q\in Y_i\cap N(u_i)$ be infected at time $\leq 1$ by $S'_{i+1}$ either it must be $y_i$ itself, which cannot happen because $y_i \in N_2(u_i)$, or $q$ must be adjacent to a vertex infected at time 0 by $S_i$, which cannot happen either because, if $q$ were adjacent to a vertex infected at time 0 by $S_i$, since it is also adjacent to $u_i$, then $q$ would be in $X_i$.
\end{itemize}
Therefore, we have that $t(G,S'_{i+1},u_i) \geq 3$. Let $S_{i+1} = S'_{i+1} \cup N_{\geq 3}(u_i)$. By Lemma \ref{lemadistanciak}, we have that $t(G,S_{i+1},u_i) \geq 3$. Thus, letting $u_{i+1} = u_i$, we have that $t(G,S_{i+1},u_{i+1}) \geq 3$.

Now, assume that every vertex in $Y_i \cap N_2(u_i)$ has exactly one neighbor in $X_i$, which is infected at time $\geq 2$ by $S_i$. Let $y_i$ be any vertex in $Y_i \cap N_2(u_i)$, let $C_i$ be the connected component of $G[Y_i]$ that contains $y_i$ and let $z_i$ be the neighbor of $y_i$ outside $C_i$, which is infected time $\geq 2$ by $S_i$. If every vertex of $C_i$ is adjacent to $z_i$, then $C_i$ has only vertices of $N_2(u)$ or only vertices of $N(u)$ because, otherwise, since $z_i$ would be equal to $u$, then there would be a vertex in $N_2(u)$ adjacent to $u$. Also, every vertex of $C_i$ has no neighbor in $N_3(u)$. Hence, either way, $z_i$ is a separator and $C_i$ would be a connected component of $G - z_i$ such that $V(C_i) \subseteq N(z_i)$ and $u \notin V(C_i)$. Therefore, $T_0$ has a vertex in $C_i$, which is a contradiction since all vertices of $C_i$ are in $Y_i$.

We then conclude that there is a vertex $y'_i$ in $C_i$ whose neighbor $z'_i$ outside $C_i$ is distinct from $z_i$. Let $S'_{i+1}= S_i\cup\{y_i\}$. All vertices in $C_i$ are infected by $S_{i+1}$ because all vertices in $C_i$ have one neighbor infected by $S_i$. Let us prove that $y'_i$ is infected by $S'_{i+1}$ at time $\geq 3$. First, we have that all vertices in $C_i$ , except for $y_i$, will be infected at time $\geq 2$ by $S'_{i+1}$ because a vertex in $V(C_i)-\{y_i\}$ can only be infected at time $\leq 1$ by $S'_{i+1}$ if it is adjacent to 2 vertices infected at time 0 by $S'_{i+1}$, which cannot happen because all vertices in $C_i$ are adjacent to only one vertex outside of $C_i$, which is infected at time $\geq 1$ by $S'_{i+1}$, and the only vertex in $C_i$ infected at time 0 by $S'_{i+1}$ is $y_i$. Additionally, since $z'_i$ is at distance $\geq 2$ of $y_i$ and $S_i$ infects $z'_i$ at time $\geq 2$, by Lemma \ref{lemadistanciak}, $z'_i$ is infected at time $\geq 2$ by $S'_{i+1}$. Therefore, since all vertices in $C_i$ are infected at time $\geq 2$ by $S'_{i+1}$ and $S'_{i+1}$ infects $z'_i$ at time $\geq 2$, then $y'_i$ is infected at time $\geq 3$ by $S'_{i+1}$. In Figure \ref{fig:gert3Ci}, we can see an example of how the infection times of the vertices in $C_i$ will be like before and after we add $y_i$ to $S_i$. 

\begin{figure}[ht]
\centering
\begin{minipage}[c]{.49\textwidth}
\begin{tikzpicture}[scale=.58]
\tikzstyle{every node}=[font=\scriptsize]
\tikzstyle{vertex}=[draw,circle,fill=black!25,minimum size=14pt,inner sep=1pt]
\tikzstyle{leafvertex}=[draw,circle,fill=black!50,minimum size=2pt,inner sep=2pt]

\def \radius {2cm}
\def \diff {2.2cm}

\node (Ci) at (90:2.5cm) {\large{\textbf{$C_i$}}};
\draw[darkgray] (0:0cm) circle (3cm);

\node[vertex] (yi) at (-18:\radius) {$y_i$};
\node[vertex] (v1) at (54:\radius) {$v_1$};
\node[vertex] (yil) at (126:\radius) {$y_i'$};
\node[vertex] (v2) at (198:\radius) {$v_2$};
\node[vertex] (v3) at (270:\radius) {$v_3$};
\node[vertex] (zi) at (-18:\radius+\diff) {$z_i$};
\node[vertex] (v4) at (198:\radius+\diff) {$v_4$};
\node[vertex] (zil) at (126:\radius+\diff) {$z_i'$};

\node[shift={(0:0.55cm)}] (nzi) at (zi) {$\geq 2$};
\node[shift={(0:0.55cm)}] (nzil) at (zil) {$\geq 2$};
\node[shift={(-30:0.55cm)}] (nv4) at (v4) {$\geq 2$};

\draw[-] (yi) to (zi);
\draw[-] (v1) to (zi);
\draw[-] (v2) to (v4);
\draw[-] (yil) to (zil);
\draw[-] (v3) to (zi);
\draw[-] (yi) to (v1);
\draw[-] (yi) to (v2);
\draw[-] (v2) to (yil);
\draw[-] (v2) to (v3);
\draw[-] (yi) to (v3);

\end{tikzpicture}
\end{minipage} 
\begin{minipage}[c]{.49\textwidth}
\begin{tikzpicture}[scale=.58]
\tikzstyle{every node}=[font=\scriptsize]
\tikzstyle{vertex}=[draw,circle,fill=black!25,minimum size=14pt,inner sep=1pt]
\tikzstyle{leafvertex}=[draw,circle,fill=black!50,minimum size=2pt,inner sep=2pt]

\def \radius {2cm}
\def \diff {2.2cm}

\node (Ci) at (90:2.5cm) {\large{\textbf{$C_i$}}};
\draw[darkgray] (0:0cm) circle (3.1cm);

\node[vertex] (yi) at (-18:\radius) {$y_i$};
\node[vertex] (v1) at (54:\radius) {$v_1$};
\node[vertex] (yil) at (126:\radius) {$y_i'$};
\node[vertex] (v2) at (198:\radius) {$v_2$};
\node[vertex] (v3) at (270:\radius) {$v_3$};
\node[vertex] (zi) at (-18:\radius+\diff) {$z_i$};
\node[vertex] (v4) at (198:\radius+\diff) {$v_4$};
\node[vertex] (zil) at (126:\radius+\diff) {$z_i'$};

\node[shift={(0:0.40cm)}] (nyi) at (yi) {$0$};
\node[shift={(0:0.55cm)}] (nv1) at (v1) {$\geq 2$};
\node[shift={(-15:0.60cm)}] (nv3) at (v3) {$\geq 2$};
\node[shift={(20:0.55cm)}](nv2) at (v2) {$\geq 3$};
\node[shift={(0:0.55cm)}] (nyil) at (yil) {$\geq 3$};
\node[shift={(0:0.60cm)}] (nzi) at (zi) {$\geq 1$};
\node[shift={(0:0.60cm)}] (nzil) at (zil) {$\geq 2$};
\node[shift={(-30:0.60cm)}] (nv4) at (v4) {$\geq 2$};

\draw[-] (yi) to (zi);
\draw[-] (v1) to (zi);
\draw[-] (v2) to (v4);
\draw[-] (yil) to (zil);
\draw[-] (v3) to (zi);
\draw[-] (yi) to (v1);
\draw[-] (yi) to (v2);
\draw[-] (v2) to (yil);
\draw[-] (v2) to (v3);
\draw[-] (yi) to (v3);

\end{tikzpicture}
\end{minipage}

\caption{Vertices of the component $C_i$ before and after the addition of $y_i$ to $S_i$.}
\label{fig:gert3Ci}
\end{figure}
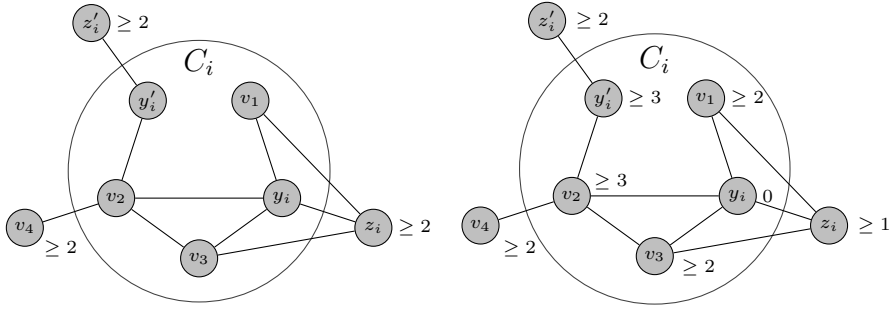

Let $S_{i+1} = S'_{i+1} \cup N_{\geq 3}(y'_i)$. By Lemma \ref{lemadistanciak}, we have that $t(G,S_{i+1},y'_i) \geq 3$. Thus, letting $u_{i+1} = y'_i$, we have that $t(G,S_{i+1},u_{i+1}) \geq 3$.

If we have that $N_2(u_i) \cap Y_i = \varnothing$, then $X_i = V(G)$. This is because $\{u_i\} \cup N_2(u_i) \cup N_{\geq 3}(u_i) \subseteq X_i$ and, for each vertex $q \in N(u_i)$, we have:
\begin{itemize}
\item If $q$ has a neighbor in $N_2(u_i)$, then $q$ is infected by this neighbor and $u_i$;

\item If $q$ has degree one, then $q \in T_0$ and, hence, $q \in X_i$;

\item If $q$ does not have a neighbor in $N_2(u_i)$ and has degree $\geq 2$, then let $C$ be the connected component of the subgraph induced by the vertices in $N(u_i)$ such that $q \in V(C)$. If $u \in V(C)$ then all vertices in $C$ are also infected, since all vertices are adjacent to $u_i \in X_i$. If $u \notin V(C)$, then either there is a vertex $v$ in $C$ that has a neighbor in $N_2(u_i)$ or every vertex in $C$ do not have a neighbor in $N_2(u_i)$. If there is a vertex $v$ in $C$ that has a neighbor in $N_2(u_i)$ then $v$ is infected by his neighbor in $N_2(u_i)$ and $u_i$ and, hence, all vertices in $C$ are in $X_i$. If every vertex in $C$ do not have a neighbor in $N_2(u_i)$ then $C$ is a connected component of the subgraph induced by $G - u_i$, where $V(C) \subseteq N(u_i)$, and $u \notin V(C)$, which means that there is a vertex $v$ in $V(C) \cap T_0$ and, thus, all vertices in $C$ are infected.

\end{itemize}
Thus, since all vertices in $N_2(u_i)$ are in $X_i$, then $X_i = V(G)$. In this case, the construction is over. The construction will eventually end. This is because, in each step, we add one vertex of $N_2(u_i)$ to $S$ and, even though the vertex $u_i$ potentially changes with each $i$, eventually $N_2(u_i) \cap Y_i = \varnothing$ because eventually there will be no more vertices of $G$ not infected. Since each step ensures that $S$ infects $u_i$ at time $\geq 3$, then, by the end of the construction, $S$ will be a hull set such that $t(G,S) \geq 3$. 
\end{proof}

Given a vertex $u$, a set $T_0 \in \mathcal{T}^u_0$ and a set of vertices $F$, such that $T_0 \cup F \cup N_{\geq 3}(u)$ infects $u$ at time 3, the construction of a hull set $S$ such that $t(G,S) \geq 3$ that is described in Lemma \ref{lema-geral3} can be implemented to run in time $O(mn^2)$. This is because, in each step $i$, since $G$ is connected, the sets $N(u_i),N_2(u_i)$ and $N_{\geq 3}(u_i)$ can be computed in time $O(m)$; finding one vertex, if there is any, belonging to the set $N_2(u_i) \cap Y_i$ that has at least one neighbor that is either not infected or infected at time $\leq 1$ by the current set $S$ can be done in time $O(n)$; and, finally, the set $Y_i$ can be updated in time $O(mn)$ (recall that $Y_i$ are the vertices not infected by $S_i$).

However, in order to decide if $t(G)\geq 3$, Lemma \ref{lema-geral3} suggests that we have to search for every set $T_0\in \mathcal{T}^u_0$, which clearly cannot be done in polynomial time. The next lemma shows us that testing just one such $T_0$ is sufficient.

\begin{lemma}
\label{lemagenept}
If there is a vertex $u$, a set $T_0 \in \mathcal{T}^u_0$ and a set of vertices $F$, with $|F|\leq 4$, such that $T_0\cup N_{\geq 3}(u)\cup F$ infects $u$ at time 3 then for every set $T_0' \in \mathcal{T}^u_0$ there is a set of vertices $F'$, where $|F'| \leq 4$, such that the set $T_0'\cup N_{\geq 3}(u)\cup F'$ infects $u$ at time 3.
\end{lemma}

\begin{proof}
It is sufficient to show that, for every $T_0'\in\mathcal{T}^u_0$ such that $T_0'-\{w'\}=T_0-\{w\}$, where $w\in T_0$ and $w'\in T'_0$, there is a set $F'$ with $|F'|\leq 4$ such that $F'\cup T_0'\cup N_{\geq 3}(u)$ also infects $u$ at time 3. By definition \ref{definicaot0u}, there is a separator $v$ of $G$ such that the vertices $w$ and $w'$ belong to the same connected component $C$ of the graph $G-v$, where $V(C)\subseteq N(v)$. If $u\neq v$, let $D$ be the connected component of $G-v$ such that $u\in V(D)$. Also, let $R=T_0\cup N_{\geq 3}(u)$ and $R'=T_0'\cup N_{\geq 3}(u)$. We can assume that $F\cap R=\varnothing$ (otherwise, let $F=F-R$).

First, suppose that $w'\in F$. Let $F'=(F-\{w'\})\cup\{w\}$. Since $F'\cup R'=F\cup R$ and $t(G,F\cup R,u)=3$, then $t(G,F' \cup R',u)=3$.

Finally, suppose that $w'\notin F$. Since $w,w'\in V(C)\subseteq N(v)$ and $C$ is a connected component, we have directly that $t(G,F\cup R',v)=t(G,F\cup R,v)$. Therefore, $t(G,F\cup R',u)=t(G,F\cup R,u)=3$. This is because, either $u=v$ (and we are done), or $u\neq v$ and consequently the infection time of any vertex not in $D$ can only affect the infection time of $u$ through $v$.
\end{proof}

\gertt*

\begin{proof}
Consider the following algorithm:\\
\begin{algorithm}[ht]
\label{algger3}
\DontPrintSemicolon
\KwIn{A connected graph $G$.}
\KwOut{\textbf{Yes}, if $t(G) \geq 3$. \textbf{No}, otherwise.}
\ForAll{$u \in V(G)$}{
	Find a set $T_0 \in \mathcal{T}^u_0$\;
	\ForAll{$F \subseteq V(G)$, where $|F| \leq 4$,}{
		\If{$F \cup T_0 \cup N_{\geq 3}(u)$ infects $u$ at time 3}{
			\Return {\textbf{Yes}}
		}
	}
}
\Return {\textbf{No}}

\caption{Algorithm that solves the Maximum Percolation Problem for $k = 3$}
\end{algorithm}

When the Algorithm \ref{algger3} outputs \textbf{Yes}, it means that there is a vertex $u$, a set $T_0 \in \mathcal{T}^u_0$, and a set $F \subseteq V(G)$, where $|F| \leq 4$, such that $F \cup T_0 \cup N_{\geq 3}(u)$ infects $u$ at time 3. Thus, by Lemma \ref{lema-geral3}, if the algorithm outputs \textbf{Yes}, then $t(G) \geq 3$. On the other hand, when the Algorithm \ref{algger3} outputs \textbf{No}, it means that, for every vertex $u$, there is a set $T_0 \in \mathcal{T}^u_0$ such that for all sets $F$, where $|F| \leq 4$, we have that $F \cup T_0 \cup N_{\geq 3}(u)$ infects $u$ at time $\neq 3$. Thus, applying the contrapositive of Lemma \ref{lemagenept}, we have that, for every vertex $u$, all sets $T_0 \in \mathcal{T}^u_0$ and all sets of vertices $F$, where $|F| \leq 4$, the set $F \cup T_0 \cup N_{\geq 3}(u)$ infects $u$ at time $\neq 3$. Therefore, by Lemma \ref{lema-geral3}, if the algorithm outputs \textbf{No} then $t(G) < 3$. Additionally, the Algorithm \ref{algger3} is an $O(mn^5)$-time algorithm because, given a vertex $u$, to find an arbitrary set in $\mathcal{T}^u_0$ and compute the set $N_{\geq 3}(u)$, it is necessary $O(m)$-time, and, given a set $F$, to test whether the set $F \cup T_0 \cup N_{\geq 3}(u)$ infects $u$ at time 3, it is necessary $O(m)$-time. Therefore, the Algorithm \ref{algger3} decides whether $t(G) \geq 3$, for any connected graph $G$, in $O(mn^5)$ time.
\end{proof}

\section{Maximum Percolation Time with parameter fixed $k=4$ in bipartite graphs}
\label{secaoBip4}

The following theorem is the main result of this section.

\begin{restatable}{theorem}{biptf}
Deciding whether $t(G)\geq 4$ is $O(m^2n^9)$-time solvable in bipartite graphs.
\end{restatable}

In order to prove this, we first define some important sets.

\begin{definition}
\label{definicaoMPQ}
Let $u,v$ be vertices of $G$. If $v\in N_{\geq 2}(u)$, define $M_v^u = P_v^u = Q_v^u = \varnothing$. If $v\in N[u]$, define $M_v^u$ to be the set of all vertices $x,y$ such that $x \neq u$, $y \neq u$, $vxy$ is an induced $P_3$, and $x$ and $y$ have degree two. Also, define $P_v^u = M_v^u \cap N(v)$ and $Q_v^u = M_v^u - P_v^u$.
\end{definition}

Note that, if $v \in N(u)$, then $P_v^u = M_v^u \cap N_2(u)$, and $Q_u^u = M_u^u \cap N_2(u)$. Also, for any $v\in N[u]$, $|P_v^u|=|Q_v^u|=|M_v^u|/2$.

\begin{definition}
\label{definicaof0u}
Let $u\in V(G)$, $v\in N(u)$ and $k_v=|P_v^u|$. Consider the sets $P_v^u=\{x_{v,1},\ldots,x_{v,k_v}\}$ and $Q_v^u=\{y_{v,1},\ldots,y_{v,k_v}\}$ such that $x_{v,i}$ and $y_{v,i}$ are adjacent for every $i\in\{1,\ldots,k_v\}$. Let $i\in\{1,\ldots,k_v\}$.
\begin{itemize}
\item Let $T'_0$ the set of all vertices in $M_v^u\cap N_2(u)$, for any $v\in N[u]$, and all vertices with degree one;
\item let $T^v_0$ the set of all vertices in $(M_{v'}^u - M_v^u) \cap N_2(u)$ for any $v'\in N[u]-\{v\}$, all vertices in $Q_v^u$ and all vertices with degree one;
\item let $T^{v,i}_0$ the set of all vertices in $(M_{v'}^u-M_v^u)\cap N_2(u)$ for any $v'\in N[u]-\{v\}$, all vertices in $(Q_v^u-\{y_{v,i}\})\cup\{x_{v,i}\}$ and all vertices with degree one.
\end{itemize}

Let $\Gamma^u_0$ be the family with $T_0'$, $T^v_0$ and $T^{v,i}_0$ for every $v\in N(u)$ and every $i\in\{1,\ldots,k_v\}$.
\end{definition}

\begin{lemma}
\label{lemacompfzu}
Given a vertex $u$, $|\Gamma^u_0|=O(m)$ and any set of the family $\Gamma^u_0$ can be computed in time $O(m)$.
\end{lemma}

\begin{proof}
First, we can obtain the set $N_2(u)$ in time $O(m)$ by breadth-first search. Following Definition \ref{definicaof0u}, we can build the set $T'_0$ in time $O(n)$ by checking, for any $q\in N_2(u)$, if $q$ has degree two and has a neighbor that also has degree two. Also, add to $T'_0$ all vertices with degree one.

For each $v\in N(u)$, we can build the set $T_0^v$ in time $O(n)$ by checking, for each vertex $q\in N_2(u)$, if $q$ has degree two and a neighbor that also has degree two. In this case, add $q$ to $T_0^v$, if $q$ and $v$ are non-adjacent; otherwise, add the neighbor of $q$ that is not $v$ to $T_0^v$. Also, add to $T_0^v$ all vertices with degree one.

For each $v\in N(u)$ and $i\in\{1,\ldots,k_v\}$, we can build the set $T^{v,i}_0$ in time $O(n)$ by checking, for each vertex $q\in N_2(u)$, if $q$ has degree two and neighbor with degree two. In this case, add $q$ to $T_0^v$, if either $q$ and $v$ are non-adjacent or $q=x_{v,i}$; otherwise, add the neighbor of $q$ that is not $v$ to $T^{v,i}_0$. Also, add to $T_0^v$ all vertices with degree one.

Therefore, since $G$ is connected and, thus, $n = O(m)$, any set in $\Gamma^u_0$ can be computed in time $O(m)$ and, since for each $v \in W$, there is $O(|N(v)|)$ sets in $\Gamma^u_0$, then $|\Gamma^u_0| = O(m)$.
\end{proof}

\begin{lemma}\label{lema-bip4}
Let $G$ be a connected bipartite graph. $t(G)\geq 4$ if and only if there is a vertex $u$, a subset $T_0\in\Gamma^u_0$ and a set of vertices $F$, where $|F|\leq 8$, such that $T_0 \cup F \cup N_{\geq 4}(u)$ infects some vertex at time 4.
\end{lemma}

\begin{proof}

First, let us prove the necessary condition. If $t(G) \geq 4$, then there is a hull set $Z$ and a vertex $u$ such that $t(G,Z,u) = 4$. By Lemma \ref{lemacoconvexo}, we have that all vertices with degree one are in $Z$. By Lemma \ref{lemadistanciak}, we have that the hull set $S'' = Z \cup N_{\geq 4}(u)$ infects $u$ at time 4. Let us prove that there is a hull set $S\supseteq S''$ which infects some vertex $x$ at time 4 and contains a set $T_0\in\Gamma_0^u$. Let us divide the proof in two cases.

The first case occurs when there is a vertex $v \in N(u)$ such that $t(G,S'',v) \geq 3$ and either $|P_v^u| \geq 2$ or $|P_v^u| = 1$ and the only vertex $x$ in $P_v^u$ is not in $S''$. Since $t(G,S'',v) \geq 3$, we have that $v$ does not have two neighbors in $S''$. Thus, if $|P_v^u| \geq 2$, by Lemma \ref{lemacoconvexo}, there is at least one vertex $x$ in $P_v^u$ infected at time $\geq 4$ by $S''$. On the other hand, if $|P_v^u| = 1$, then the only vertex $x$ in $P_v^u$ is not in $S''$ and, by Lemma \ref{lemacoconvexo}, we have that $x$ is infected at time $\geq 4$ by $S''$. At any case, there is a vertex $x \in P_v^u$ infected at time $\geq 4$ by $S''$. By Lemma \ref{lemainfectaexato}, we have that there is a hull set $S' \supseteq S''$ that infects $x$ at time 4 and $v$ at time 3. Recall that $v \in N(u)$ and $x \in N_2(u)$. Let $S = S' \cup (N_2(u) \cap \bigcup_{v'\in \{u\} \cup (N(u)-\{v\})} (M_{v'}^u - M_{v}^u))$. Note that, since $v$ is infected by $S''$ at time 3 and $x$ at time 4, then $v$ must have degree at least 3. Therefore, $v \notin M_u^u$ and, thus, $x \notin M_u^u$. Since all vertices in $S - S'$ are in $N_2(u)$, $G$ is bipartite and $v \notin S - S'$, then each vertex $w \in S - S'$ is either at distance 2 or at distance $\geq 4$ of $x$. If $w$ is at distance 2 of $x$, then there is a vertex $z$ adjacent to both $w$ and $x$ and, since $z$ cannot be $v$, because $w$ would be in $M_v^u$, and $x$ has degree two, then $z$ is in $Q_v^u$ and also in $S'$, which means that $w$ is at distance $\geq 3$ of $v$. Therefore, when we add $w$ to $S'$ to form the set $S$, we have that the resulting set still infects $z$ at time 0 and, by Lemma \ref{lemadistanciak}, $v$ at time 3 and, thus, infects $x$ at time 4. On the other hand, if $w$ is at distance $\geq 4$ of $x$, and, therefore, at distance $\geq 3$ of $v$, by Lemma \ref{lemadistanciak} applied for each vertex of $S-S'$, we have that $S$ infects $v$ at time 3 and $x$ at time 4.

The second case occurs when, for every vertex $v\in N(u)$, $P_v^u = \varnothing$ or $t(G,S'',v)$ $\leq 2$ or $|P_v^u| = 1$ and the only vertex $x$ in $P_v^u$ is in $S''$. Let $S = S'' \cup (N_2(u) \cap \bigcup_{v \in N[u]} M_{v}^u)$. Then, for each vertex $w \in S - S''$, we have that either $w$ is already in $S''$ or, if $w$ is not in $S''$, then either $w$ is in $P_v^u$ for some vertex $v \in N(u)$ or $w$ is in $Q_u^u$. 

If $w$, which is not in $S''$, is in $P_v^u$ for some vertex $v \in N(u)$, then, since $|P_v^u| > 0$ and, if $|P_v^u| = 1$ we would have that the only vertex in $P_v^u$ would be not in $S''$, $v$ is infected at time $\leq 2$ by $S''$. Note that, in this case, $v$ is the only vertex in $N(u)$ infected at time $\leq 2$ by $S''$. Thus, for each vertex $v' \in N(u) - \{v\}$, since $G$ is bipartite and $w$ has degree two, we have either $w$ is at distance $\geq 3$ of $v'$ or $w$ is adjacent to $v'$. But, in fact, we have that $w$ cannot be adjacent to $v'$. This is because, if $w$ is adjacent to $v'$, then $v'$ has degree two, because, since $w \in P_v^u$, $v'$ would be in $Q_v^u$, and is adjacent only to $u$ and $w$. Therefore, since $w \notin S''$, then, by Lemma \ref{lemacoconvexo}, $v' \in S''$, and, therefore, since $v'$ and $v$ are infected respectively at times 0 and 2 by $S''$, then $u$ would be infected by $S''$ at time $\leq 3$. Therefore $w$ is at distance $\geq 3$ of $v'$ and, since $v$ is infected at time $\leq 2$ and $u$ at time 4 by $S''$, then $v'$ is infected at time $\geq 3$ by $S''$, and, therefore, by Lemma \ref{lemadistanciak}, when we add $w$ to $S''$ to form $S$, the resulting set infects all vertices in $N(u)$ that are infected at time $\geq 3$ by $S''$ at time $\geq 3$. 

On the other hand, if $w$, which is not in $S''$, is in $Q_u^u$, then $w$ cannot be adjacent to any vertex in $N(u)-\{z\}$, where $z$ is the only vertex in $N(w) \cap P_u^u$. This is because, since $w \notin S''$, by Lemma \ref{lemacoconvexo}, then $z \in S''$. Therefore, since $S''$ infects $u$ at time 4, all vertices in $N(u)-\{z\}$ are infected at time $\geq 3$ by $S''$. But, if $w$ is adjacent to some vertex $v \in N(u)-\{z\}$, since $w \in P_v^u$ and $z \in Q_v^u$, then we would have a vertex $v \in N(u)$ such that $t(G,S'',v) \geq 3$ and either $|P_v^u| \geq 2$ or $|P_v^u| = 1$ and the only vertex $w$ in $P_v^u$ is not in $S''$, which falls in the first case. Therefore, since $G$ is bipartite and $w$ is not adjacent to any vertex in $N(u)-\{z\}$, then $w$ is at distance $\geq 3$ of all vertices in $N(u)-\{z\}$. Therefore, by Lemma \ref{lemadistanciak}, when we add $w$ to $S''$ to form $S$, the resulting set infects all vertices of $N(u)$ that are infected at time $\geq 3$ by $S''$ at time $\geq 3$.   

Hence, in any of the two cases, we have that $S$ infects all vertices in $N(u)$, except at most by one, at time $\geq 3$ and, therefore, $S$ infects $u$ at time $\geq 4$. Also, by Lemma \ref{corolariosubset}, $S$ infects $u$ at time $\leq 4$ and, therefore, $S$ infects $u$ at time 4. Also, we have that, by Definition \ref{definicaof0u}, there is a hull set $S$ and a set $T_0 \in \Gamma_0^u$ such that $T_0 \subseteq S$.

Since $S$ infects some vertex $x$ at time 4 and, to be infected, each vertex needs only two neighbors infected, then there is a set $F \subseteq S$ that infects $x$ at time 4 such that $|F| \leq 2^4 = 16$. Thus, since there is a set $T_0 \in \Gamma^u_0$ such that $T_0 \subseteq S$ and $N_{\geq 4}(u) \subseteq S$, then $F \subseteq T_0 \cup N_{\geq 4}(u) \cup F \subseteq S$. By Lemma \ref{corolariosubset}, we have that $4 = t(G,S,x) \leq t(G,T_0 \cup N_{\geq 3}(u) \cup F,x) \leq t(G,F,x) = 4$. Therefore, there is a set $T_0 \in \Gamma^u_0$, such that the set $T_0 \cup N_{\geq 4}(u) \cup F$, with $|F| \leq 16$, infects the vertex $x$ at time 4. Furthermore, we can decrease the size of $F$ to 8 by the following claim, which is proved in Section \ref{apendClaim}.

\begin{restatable}{claim}{claimbmf}\label{claimbmf}
Let $G$ be a bipartite graph and $u,x\in V(G)$.
If $G$ has a hull set which contains a set $T_0\in\Gamma_0^u$ and infects $x$ at time 4, then there exists a set $F$ with at most 8 vertices and a hull set $S$ such that $F\subseteq S$ and $t(G,F,x)=t(G,S,x)=4$.
\end{restatable}

Therefore, $T_0\cup N_{\geq 4}(u)\cup F$ infects the vertex $x$ at time 4 and we conclude the proof of the necessary condition.

Now, let us prove the sufficient condition. Suppose that there is a vertex $u$, a set $F$, with $|F| \leq 8$, and a set $T_0 \in \mathcal{T}^u_0$ such that the set $S_0 = T_0 \cup N_{\geq 4}(u) \cup F$ infects some vertex $x$ at time 4. We will show how to construct a hull set $S$ such that $t(G,S) \geq 4$. We begin with $S=S_0$. Let $S_i$ be the constructed set at step $i$. Each step adds one vertex to $S$ and, at the end of each step, it is guaranteed that $S_i$ infects, at first, $x$ at time 4 and, in the last step, some vertex at time $\geq 4$. Let $X_i$ be the set of vertices infected by $S_i$, and let $Y_i = V(G) - X_i$. It is worth noting that a vertex in $Y_i$ is adjacent to at most one vertex in $X_i$ and, since it must have degree at least two, it is also adjacent to at least one vertex in $Y_i$. If $S_0$ is a hull set, we stop the construction because, since $t(G,S_0,x) = 4$, then $t(G,S_0) \geq 4$.

If not, for $i \geq 0$, in the $(i+1)$th step of the construction, suppose that there exists a vertex $y_i \in Y_i \cap N_2(x)$ with no neighbor infected by $S_i$ at time $\geq 2$. Let $S_{i+1}=S_i\cup\{y_i\}$. Clearly, $x$ has at most one neighbor infected by $S_i$ at time $\leq 2$ and, by the choice of $y_i$, $x$ is not adjacent to $y_i$. Let us prove that $t(G,S_{i+1},x) = 4$. We have that:
\begin{itemize}
\item For every vertex $v \in N(x) \cap N(y_i)$ that is not infected by $S_i$, we have that $t(G,S_{i+1},v)$ $\geq 3$. This is because, since every neighbor of $v$, except $x$, is not infected by $S_i$ and, since $G$ is bipartite, every neighbor of $v$ is at distance $\geq 2$ of $y_i$, then, by Lemma \ref{lemadistanciak}, for every vertex $s\in N(v)-\{x\}$, $t(G,S_{i+1},s)\geq 2$ and, thus, $t(G,S_{i+1},v)\geq 3$.
\item For every vertex $v \in N(x) \cap N_{\geq 3}(y_i)$ such that $t(G,S_i,v) \geq 3$, by Lemma \ref{lemadistanciak}, we have that $t(G,S_{i+1},v) \geq 3$. 
\end{itemize}
Additionally, there is at most one vertex in $N(x)$ infected at time $\leq 2$ by $S_i$ and, since $y_i \in N_2(x)$ and $G$ is bipartite, $N(x) \cap N_2(y_i) = \varnothing$. Thus, we have that $t(G,S_{i+1},x) \geq 4$. Also, since $S_i \subseteq S_{i+1}$, by Lemma \ref{corolariosubset}, we have that $t(G,S_{i+1},x) \leq 4$. Therefore $t(G,S_{i+1},x) = 4$.

If all the vertices in the set $Y_i\cap N_2(x)$ have a neighbor infected by $S_i$ at time $\geq 2$, suppose that there is a vertex $y_i \in Y_i\cap N_3(x)$ with no neighbor infected by $S_i$ at time $\geq 2$. Let $S_{i+1}=S_i\cup\{y_i\}$. Let us prove that $t(G,S_{i+1},x) = 4$. Let $v$ be any vertex in $N(x)$ such that $t(G,S_i,v) \geq 3$. We have that:
\begin{itemize}
\item For every vertex $s \in N(v) \cap N(y_i)$ that is not infected by $S_i$, we have that $t(G,S_{i+1},s)$ $\geq 2$. This is because, since $s$ is in $N_2(x)$, then every neighbor of $s$, except one, say $q$, which is infected at time $\geq 2$ by $S_i$, is not infected by $S_i$ and, additionally, every neighbor of $s$ is at distance $\geq 2$ of $y_i$, then, by Lemma \ref{lemadistanciak}, for every vertex $r \in N(s)$, $t(G,S_{i+1},r)$ $\geq 2$ and, thus, $t(G,S_{i+1},s)$ $\geq 3 \geq 2$.
\item For every vertex $s \in N(v) \cap N_{\geq 3}(y_i)$ such that $t(G,S_i,s)$ $\geq 2$, by Lemma \ref{lemadistanciak}, we have that $t(G,S_{i+1},s)$ $\geq 2$. 
\end{itemize}

Additionally, there is at most one vertex in $N(v)$ infected at time $\leq 1$ by $S_i$ and, since $y_i \in N_3(x)$ and $G$ is bipartite, $N(v) \cap N_2(y_i) = \varnothing$. Thus, we have that $t(G,S_{i+1},v) \geq 3$ and, therefore, since there is at most one vertex in $N(x)$ infected at time $\leq 2$ by $S_i$, $t(G,S_{i+1},x) \geq 4$. Also, since $S_i \subseteq S_{i+1}$, by Lemma \ref{corolariosubset}, we have that $t(G,S_{i+1},x) \leq 4$. Therefore $t(G,S_{i+1},x) = 4$. It is worth noting that $y_i$ is at distance $\geq 2$ of every vertex that is infected at time $\geq 2$ by $S_i$, which, by Lemma \ref{lemadistanciak}, implies that it is not possible to go back to the previous state, i.e., it is not possible that there is a vertex in $Y_{i+1} \cap N_2(x)$ with no neighbor infected by $S_{i+1}$ at time $\geq 2$.

When all vertices in the set $Y_i\cap N_2(x)$ and in the set $Y_i\cap N_3(x)$ have a neighbor infected by $S_i$ at time $\geq 2$, let $C_i$ be any connected component of $G[Y_i]$. We have that every vertex of $C_i$ has exactly one neighbor outside $C_i$, which is infected at time $\geq 2$ by $S_i$. 

We have that $C_i$ has at least 3 vertices because, otherwise, we would have in $C_i$ either only one vertex of degree one, which would be in $T_0$, or two adjacent vertices of degree two, that are not $u$ because $u \in X_i$. If the latter happens, we have that at least one of the two vertices in $C_i$ would be in $S_i$ because of the following:

\begin{restatable}{claim}{claimtduplo}\label{claimtduplo}
For all pair of neighbor vertices $x,y \neq u$ that have degree 2, we have that $\{x,y\} \cap (T_0 \cup N_{\geq 4}(u)) \neq \varnothing$.
\end{restatable}

\begin{proof}
Since $G$ is bipartite, we can assume, without loss of generality, that $x \in N_k(u)$, $y \in N_{k+1}(u)$. If $k=3$, then $y \in N_{\geq 4}(u)$. So, assume that $1 \leq k \leq 2$. Since $x \in N_k(u)$, then there a vertex $v \in N_{k-1}(u)$ that is neighbor of $x$. Thus, $x \in P_v^u$ and $y \in Q_v^u$. Since, by Definition \ref{definicaof0u}, for each vertex $v \in N[u]$, $w \in P_v^u$ and $z \in Q_v^u$, $T_0$ has either $w$ or $z$, then either $x$ or $y$ is in $T_0$.
\end{proof}

Therefore, in either case, one vertex of $C_i$ would also be in $T_0$ and, consequently, in $X_i$, which would be a contradiction. Thus, since the graph is bipartite, there are two vertices $y_i$ and $y_i'$ in $V(C_i)$ that are at distance 2 of each other. Let $S_{i+1}' = S_i \cup \{y_i\}$. Let us prove that $S_{i+1}'$ infects $y_i'$ at time $\geq 4$. Let $z_i$ and $z_i'$ be the neighbors of $y_i$ and $y_i'$ outside of $C_i$ respectively. Also, let $v$ be any vertex in $N(y_i) \cap V(C_i)$ and $v'$ his neighbor outside $C_i$. We have that, since $G$ is bipartite, $y_i$ cannot be neighbor of $v'$ and, since $v'$ is infected by $S_i$ at time $\geq 2$ then, by Lemma \ref{lemadistanciak}, $S_{i+1}'$ infects $v'$ at time $\geq 2$. Also, since all vertices in $N(v) \cap V(C_i)$ are at distance $\geq 2$ of $y_i$, except $y_i$ itself, by Lemma \ref{lemadistanciak}, they all are infected at time $\geq 2$ by $S_{i+1}'$. Thus, $S_{i+1}'$ infects $v$ at time $\geq 3$. Therefore, since all neighbors of $y_i'$ in $V(C_i)$ are either at distance $\geq 3$ of $y_i$, and then, by Lemma \ref{lemadistanciak}, are infected by $S_{i+1}'$ at time $\geq 3$, or at distance 1 of $y_i$, which are also infected at time $\geq 3$ by $S_{i+1}'$, we have that $S_{i+1}'$ infects at most one neighbor of $y_i'$ at time $\leq 2$, which is $z_i'$. Since all vertices in $C_i$ are infected by $S_{i+1}'$, then $y_i'$ is infected by $S_{i+1}'$ at time $\geq 4$. In Figure \ref{fig:bipt4Ci}, we can see an example of how the infection times of the vertices in $C_i$ will be like before and after we add $y_i$ to $S_i$.

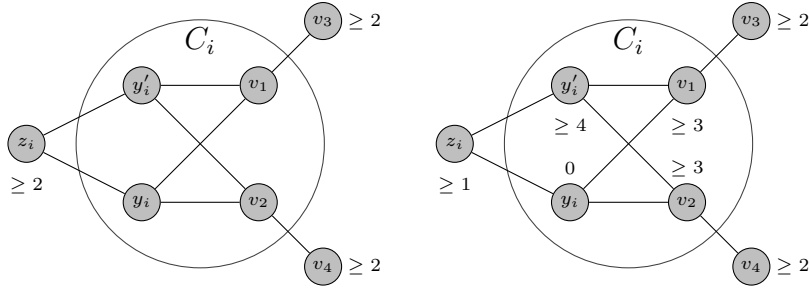
\begin{figure}[ht]
\centering
\begin{minipage}[c]{.46\textwidth}
\begin{tikzpicture}[scale=.55]
\tikzstyle{every node}=[font=\scriptsize]
\tikzstyle{vertex}=[draw,circle,fill=black!25,minimum size=14pt,inner sep=1pt]
\tikzstyle{leafvertex}=[draw,circle,fill=black!50,minimum size=2pt,inner sep=2pt]

\def \radius {2cm}
\def \diff {2.2cm}

\node (Ci) at (90:2.5cm) {\large{\textbf{$C_i$}}};
\draw[darkgray] (0:0cm) circle (3cm);

\node[vertex] (yi) at (225:\radius) {$y_i$};
\node[vertex] (v1) at (45:\radius) {$v_1$};
\node[vertex] (yil) at (135:\radius) {$y_i'$};
\node[vertex] (v2) at (315:\radius) {$v_2$};

\node[vertex] (zi) at (180:\radius+\diff) {$z_i$};
\node[vertex] (v3) at (45:\radius+\diff) {$v_3$};
\node[vertex] (v4) at (315:\radius+\diff) {$v_4$};

\node[shift={(-90:0.55cm)}] (nzi) at (zi) {$\geq 2$};
\node[shift={(0:0.55cm)}] (nv4) at (v4) {$\geq 2$};
\node[shift={(0:0.55cm)}] (nv3) at (v3) {$\geq 2$};

\draw[-] (yi) to (v1);
\draw[-] (yi) to (v2);
\draw[-] (yil) to (v1);
\draw[-] (yil) to (v2);

\draw[-] (yi) to (zi);
\draw[-] (yil) to (zi);
\draw[-] (v1) to (v3);
\draw[-] (v2) to (v4);

\end{tikzpicture}
\end{minipage} 
\begin{minipage}[c]{.46\textwidth}
\begin{tikzpicture}[scale=.55]
\tikzstyle{every node}=[font=\scriptsize]
\tikzstyle{vertex}=[draw,circle,fill=black!25,minimum size=14pt,inner sep=1pt]
\tikzstyle{leafvertex}=[draw,circle,fill=black!50,minimum size=2pt,inner sep=2pt]

\def \radius {2cm}
\def \diff {2.2cm}

\node (Ci) at (90:2.5cm) {\large{\textbf{$C_i$}}};
\draw[darkgray] (0:0cm) circle (3cm);

\node[vertex] (yi) at (225:\radius) {$y_i$};
\node[vertex] (v1) at (45:\radius) {$v_1$};
\node[vertex] (yil) at (135:\radius) {$y_i'$};
\node[vertex] (v2) at (315:\radius) {$v_2$};

\node[vertex] (zi) at (180:\radius+\diff) {$z_i$};
\node[vertex] (v3) at (45:\radius+\diff) {$v_3$};
\node[vertex] (v4) at (315:\radius+\diff) {$v_4$};

\node[shift={(90:0.45cm)}] (nyi) at (yi) {$0$};
\node[shift={(-90:0.55cm)}] (nyil) at (yil) {$\geq 4$};
\node[shift={(-90:0.55cm)}] (nv1) at (v1) {$\geq 3$};
\node[shift={(90:0.45cm)}] (nv2) at (v2) {$\geq 3$};

\node[shift={(-90:0.55cm)}] (nzi) at (zi) {$\geq 1$};
\node[shift={(0:0.55cm)}] (nv4) at (v4) {$\geq 2$};
\node[shift={(0:0.55cm)}] (nv3) at (v3) {$\geq 2$};

\draw[-] (yi) to (v1);
\draw[-] (yi) to (v2);
\draw[-] (yil) to (v1);
\draw[-] (yil) to (v2);

\draw[-] (yi) to (zi);
\draw[-] (yil) to (zi);
\draw[-] (v1) to (v3);
\draw[-] (v2) to (v4);

\end{tikzpicture}
\end{minipage}

\caption{Vertices of the component $C_i$ before and after the addition of $y_i$ to $S_i$.}
\label{fig:bipt4Ci}
\end{figure}

If $y_i'$ is in $N_2(x)$ (resp. $N(x)$ or $N_3(x)$), let $S_{i+1} = S_{i+1}' \cup ((Y_i - V(C_i)) \cap N_2(x))$ (resp. $S_{i+1} = S_{i+1}' \cup ((Y_i - V(C_i)) \cap (N(x) \cup N_3(x)))$). Since $G$ is bipartite, then each vertex in $S_{i+1} - S_{i+1}'$ are at distance $\geq 3$ of all the neighbors of $y_i'$ in $C_i$ and, therefore, since all neighbors of $y_i'$ in $C_i$ are infected at time $\geq 3$ by $S_{i+1}'$, applying Lemma \ref{lemadistanciak} once for each vertex in $S_{i+1} - S_{i+1}'$, we have that $S_{i+1}$ infects all neighbors of $y_i'$ in $C_i$ at time $\geq 3$. Thus, since $y_i'$ has only one neighbor outside $C_i$, then $S_{i+1}$ infects $y_i'$ at time $\geq 4$. Also, since there is at least one vertex in $N_2(x)$ and one vertex in either $N(x)$ or $N_3(x)$ in each connected component of $G[Y_i]$, because each connected component of $G[Y_i]$ has at least three vertices, then all connected components of $G[Y_i]$ are infected by $S_{i+1}$ and, thus, $S_{i+1}$ is a hull set that infects $G$ at time $\geq 4$. At this point, we stop the construction and we have the constructed set $S$ is a hull set such that $t(G,S) \geq 4$.
\end{proof}

Given a vertex $u$, a set $T_0\in\Gamma^u_0$ and a set of vertices $F$, such that $T_0\cup F\cup N_{\geq 3}(u)$ infects some vertex $x$ at time 4, the construction to build a hull set $S$ such that $t(G,S)\geq 4$ that is described in Lemma \ref{lema-bip4} can be implemented to run in time $O(mn^2)$. This is because, since $G$ is connected, the sets $N(u),N_2(u),N_3(u)$ and $N_{\geq 4}(u)$ can be computed in time $O(m)$. Also, for each step $i$, finding one vertex, if there is any, belonging to the set $N_2(u) \cap Y$ and the set $N_3(u) \cap Y$ that has at least one neighbor that is either not infected or infected at time $\leq 1$ by the current set $S$ can be done in time $O(n)$ and updating the set $Y$ can be done in time $O(mn)$. Finally, the last step can be done in time $O(m)$.

Now, let us prove the main theorem of the section.

\biptf*

\begin{proof}
Consider the following algorithm:\\
\begin{algorithm}[ht]
\label{algbip4}
\DontPrintSemicolon
\KwIn{A bipartite connected graph $G$.}
\KwOut{\textbf{Yes}, if $t(G) \geq 4$. \textbf{No}, otherwise.}
\ForAll{$u \in V(G)$}{
	\ForAll{$T_0 \in \Gamma_0^u$}{
		\ForAll{$F \subseteq V(G)$, where $|F| \leq 8$,}{
			\If{$F \cup T_0 \cup N_{\geq 4}(u)$ infects any vertex at time 4}{
				\Return {\textbf{Yes}}
			}
		}
	}
}
\Return {\textbf{No}}

\caption{Algorithm that solves the Maximum Percolation Problem for $k = 4$ in bipartite graphs}
\end{algorithm}

The Algorithm \ref{algbip4} outputs \textbf{Yes} if and only if, given a bipartite graph $G$, there is a vertex $u$, a set $T_0 \in \Gamma^u_0$, and a set $F \subseteq V(G)$, where $|F| \leq 8$, such that $F \cup T_0 \cup N_{\geq 4}(u)$ infects some vertex at time 4. Thus, by Lemma \ref{lema-bip4}, the Algorithm \ref{algbip4} outputs \textbf{Yes} if and only if $t(G) \geq 4$. Furthermore, we have that, given a vertex $u$, by Lemma \ref{lemacompfzu}, there is $O(m)$ sets in $\Gamma_0^u$ and each one can be computed at time $O(n^2)$. Additionally, given a set $F$, to test whether the set $F \cup T_0 \cup N_{\geq 4}(u)$ infects some vertex at time 4, it is necessary $O(m)$-time. Thus, the Algorithm \ref{algbip4} runs at time $O(n \cdot (m \cdot (m + n^8 \cdot m))) = O(m^2n^9)$. Therefore, the Algorithm \ref{algbip4} decides whether $t(G) \geq 4$, for any bipartite connected graph $G$, in $O(m^2n^9)$ time.
\end{proof}

\renewcommand{\thelemma}{\Alph{section}.\arabic{lemma}}
\renewcommand{\thecorollary}{\Alph{section}.\arabic{lemma}.\arabic{corollary}}
\renewcommand{\thefact}{\arabic{fact}}

\section{Technical lemmas}
\label{apendClaim}

\subsection{Proof of Claim \ref{claimmf}}

\claimmf*
\begin{proof}
First, let us start by defining the set $C_{H,v,t}$. For any vertex $v$ infected by some hull set $H$ at time $t$, let us define $C_{H,v,t}$ to be any minimal subset of $H$ such that the set $C_{H,v,t}$ infects $v$ at time $t$. In general, if $v$ is infected at time $t$ by $H$, then $|C_{H,v,t}| \leq 2^t$ because every vertex needs only two infected neighbors to be infected at any time $\geq 1$. Additionally, if $v$ is infected at time $t \geq 1$ and has a neighbor vertex infected at time 0 by $H$, since $v$ must have a neighbor $v'$ infected at time $t-1$, then, we have that $|C_{H,v,t}| \leq |C_{H,v',t-1}| + 1 \leq 2^{t-1} + 1$. Particularly, if $v$ is infected at time 2 and has a neighbor infected at time 0 by $H$, then $|C_{H,v,t}| \leq 3$.\\

We will start the proof by demonstrating the following fact:

\begin{fact}
\label{fato1}
There is a hull set $S'\supseteq S$ and a vertex $v\in N(u)$ infected at time 2 by $S$, such that $u$ is also infected at time 3 by $S'$, $v$ is also infected at time 2 by $S'$ and $v$ has a neighbor in $S'$.
\end{fact}

\begin{proof}
Let $D$ be the set of vertices in $N(u)$ that are infected at time 2 by $S$. If there is a vertex $v \in D$ that has a neighbor in $S$, then just let $S' = S$ and we are done.

If not, then no vertex in $D$ has a neighbor in $S$. Let $v$ be any vertex in $D$, and $w$ be a vertex in $N(v) \cap N_2(u)$ that is infected at time 1 by $S$. Let $S' = S \cup \{w\}$. We have that:
\begin{itemize}
\item For every vertex $z$ in $D$, $z$ does not have neighbors in $S$. Thus, $z$ will have, at most, one neighbor in $S'$ because $S' = S \cup \{w\}$. Therefore, $z$ is infected at time $\geq 2$ by $S'$.

\item For every vertex $z$ adjacent to $u$ and infected at time $t$ by $S$, where $t \geq 3$, we have that $z$ cannot be adjacent to $w$ and another vertex in $S$ because, otherwise, it would be infected at time $\leq 2$ by $S$. Therefore, $z$ has at most one neighbor in $S'$ and, hence, is infected by $S'$ at time $\geq 2$.
\end{itemize}

Thus, we have that $t(G,S',u) \geq t(G,S,u)$, and, since $S \subseteq S'$, we also have that, by Lemma \ref{corolariosubset}, $t(G,S',u) \leq t(G,S,u)$. Therefore, $t(G,S',u) = t(G,S,u) = 3$. Additionally, we have that the vertex $v$, which is infected at time 2 by $S'$, as we have shown before, has a neighbor in $S'$. Since $S \subseteq S'$, then we have that $S'$ is the hull set we are looking for.
\end{proof}

Thus, let $S'$ be the hull set as described in the Fact \ref{fato1}. Since $v$ has one neighbor in $S'$, we have that $|C_{S',v,2}| \leq 3$.

Suppose that there is a vertex $w$ in $N(u) \cap S'$, let $F = C_{S',v,2} \cup \{w\}$ and $R = S'$. Since $t(G,F,u) \leq 3$, $|F| \leq 4$, $F \subseteq R$ and $t(G,R,u) = 3$, then $F$ and $R$ are the sets that we are looking for. 

Now, suppose that there is a vertex $w$ in $N(u)$ such that $t(G,S',w) = 1$. If $w$ is adjacent to some vertex $z \in N(u)$ such that $z$ has a neighbor $s$ in $S'$, then let $F = C_{S',w,1} \cup \{s\}$ and $R = S'$. Since $t(G,F,s) = 0$ and $t(G,F,w) = 1$ then $t(G,F,z) \leq 2$ and, therefore, $t(G,F,u) \leq 3$. Since $t(G,F,u) \leq 3$, $|F| \leq 4$, $F \subseteq R$ and $t(G,R,u) = 3$, then $F$ and $R$ are the sets that we are looking for. However, if $w$ has no neighbor in $N(u)$ that has a neighbor in $S'$, let $F = C_{S',v,2} \cup \{w\}$ and $R = S' \cup \{w\}$. We have that $t(G,F,u) \leq 3$ because $t(G,F,v) \leq 2$ and $t(G,F,w) = 0$. Additionally, since all vertices in $N(u)-\{w\}$ are infected by $S'$ at time $\geq 2$ and all vertices in $N(u)$ that are adjacent to $w$ are not adjacent to any vertex in $S'$, then all vertices in $N(u)-\{w\}$ have at most one neighbor in $R$. Thus, all vertices in $N(u)-\{w\}$ are infected at time $\geq 2$. Therefore, we have that $t(G,R,u) \geq 3$, and, since $t(G,F,u) \leq 3$, $|F| \leq 4$, $F \subseteq R$ and $t(G,R,u) \geq 3$, then $F$ and $R$ are the sets that we are looking for. 

Henceforth, we will consider only the cases where all vertices in $N(u)$ are infected at time $\geq 2$ by $S'$.

Suppose that there is a vertex $v'$ in $N(u)-\{v\}$ that has no neighbors in $N(u)$ that has some neighbor in $S'$. In this case, let $F = C_{S',v,2} \cup \{v'\}$ and $R = S' \cup \{v'\}$. We have that $t(G,F,u) \leq 3$ because $t(G,F,v) \leq 2$ and $t(G,F,v') = 0$. Additionally, since all vertices in $N(u)-\{v'\}$ are infected by $S'$ at time $\geq 2$ and all vertices in $N(u)$ that are adjacent to $v'$ are not adjacent to any vertex in $S'$, then all vertices in $N(u)-\{v'\}$ have at most one neighbor in $R$. Since all vertices in $N(u)$ are infected by $S'$ at time $\geq 2$, then all vertices in $N(u)-\{v'\}$ are infected by $R$ at time $\geq 2$. Therefore, we have that $t(G,R,u) \geq 3$, and, since $t(G,F,u) \leq 3$, $|F| \leq 4$, $F \subseteq R$ and $t(G,R,u) \geq 3$, then $F$ and $R$ are the sets that we are looking for. 

From now on, suppose that every vertex in $N(u)-\{v\}$ has a neighbor vertex in $N(u)$ that has a neighbor in $S'$. Let us show the following fact:

\begin{fact}
\label{fato2}
There is a hull set $S'' \supseteq S'$ that infects $u$ at time 3 and there is a vertex $v' \in N(u) - \{v\}$ that is infected at time 2 by $S'$ such that $v'$ is infected at time 2 by $S''$ and has a neighbor $w$ in $S''$.
\end{fact}

\begin{proof}
This proof is similar to the proof of the Fact \ref{fato1}. If there is a vertex $v'$ in $N(u)-\{v\}$ that is infected at time 2 by $S'$ and that has a neighbor in $S'$, then, letting $S'' = S'$, $S''$ is the set we are looking for. Then, let us assume that there is no vertex in $N(u) - \{v\}$ that is infected at time 2 by $S'$ and that has a neighbor in $S'$. Then, since all vertices in $N(u)$ are infected by $S'$ at time $\geq 2$ and $t(G,S',u) = 3$, there must be at least one vertex in $N(u) - \{v\}$ infected at time 2 by $S'$. Let $v'$ be such vertex. Let $S'' = S' \cup \{w\}$, where $w$ is in $N(v')$ and is infected at time 1 by $S'$. We have that $w \in N_2(u)$ because, otherwise, there would be a vertex infected by $S'$ at time 1 in $N(u)$. Let $D$ be the set of vertices in $N(u)$ that are infected at time 2 by $S'$. We have that:

\begin{itemize}
\item For every vertex $z$ in $D - \{v\}$, $z$ does not have neighbors in $S'$. Thus, $z$ will have at most one neighbor in $S''$, because $S'' = S' \cup \{w\}$, and, hence, $t(G,S'',z) \geq 2$.

\item For every vertex $z$ infected at time $t$, where $t \geq 3$, by $S'$, we have that $z$ cannot be adjacent to $w$ and another vertex in $S'$ because, otherwise, it would be infected at time $\leq 2$ by $S'$. Therefore, $z$ has at most one neighbor in $S''$ and, hence, is infected by $S''$ at time $\geq 2$.
\end{itemize}

Tus, we have that $v$ is the only vertex that can be infected at time $\leq 1$ by $S''$. Also, since all vertices in $N(u)-\{v\}$ are infected at time $\geq 2$ by $S''$, we have that $t(G,S'',u) \geq 3$. Since $S' \subseteq S''$, then, by Lemma \ref{corolariosubset}, we also have that $t(G,S'',u) \leq 3$ and, thus, $t(G,S'',u) = 3$. Also, by Lemma \ref{corolariosubset}, $t(G,S'',v') = 2$. Then, since $t(G,S'',w) = 0$, $S''$ is the set that we are looking for.

\end{proof}

Thus, suppose that $v$ is infected at time 1 by $S''$. Since $v$ is infected at time $2$ by $S'$ and 1 by $S''$, then $v$ must be adjacent to $w$ and to some other vertex $r$, which is in $S'$ and, therefore, it is also in $S''$. Let $x$ be a vertex in $N(v')$ that is infected at time 1 by $S''$. Let $F = \{r,w\} \cup C_{S'',x,1}$ and $R = S''$. Since $t(G,F,w) = 0$ and $t(G,F,x) = 1$, then $t(G,F,v') = 2$. Hence, $t(G,F,u) \leq 3$ because $t(G,F,v') \leq 2$ and, since $t(G,F,r) = 0$ and $t(G,F,w) = 0$, $t(G,F,v) = 1$. Thus, since $t(G,F,u) \leq 3$, $|F| \leq 4$, $F \subseteq R$ and $t(G,R,u) = 3$, then $F$ and $R$ are the sets that we are looking for.

Thus, henceforth, suppose that $v$ is not infected at time $\leq 1$ by $S''$. Since, in the proof of the Fact \ref{fato2}, we showed that $v$ was the only vertex that could be infected by $S''$ at time $\leq 1$, then, there is no vertex in $N(u)$ infected at time $\leq 1$ by $S''$. Let $x$ be a vertex in $N_2(u) \cap N(v')$ that is infected at time 1 by $S''$.

Also, if there is a vertex $z \in N(u) - \{v'\}$ that is adjacent to $x$ and to some vertex in $q \in S''$, we have that $z$ is infected at $\leq 2$ by $S''$. Since, we assumed previously that there is no vertex in $N(u)$ that is infected at time $\leq 1$ by $S''$, we have that $z$, in fact, is infected at time 2 by $S''$. Let $F = \{q,w\} \cup C_{S'',x,1}$ and $R = S''$. We have that $t(G,F,z) \leq 2$ because $t(G,F,q) = 0$ and $t(G,F,x) \leq 1$. Also, $t(G,F,v') \leq 2$ because $t(G,F,w) = 0$ and $t(G,F,x) \leq 1$. Therefore, $t(G,F,u) \leq 3$. Since, we have that $q,w \in S''$, then $F \subseteq R$, and, therefore, since $t(G,F,u) \leq 3$, $|F| \leq 4$, $F \subseteq R$ and $t(G,R,u) = 3$, then $F$ and $R$ are the sets that we are looking for.

If, however, there is no vertex in $N(u)-\{v'\}$ that is adjacent to $x$ and to some vertex in $S''$, let $R = S'' \cup \{x\}$. Since there is no vertex in $N(u) \cap S''$ and there is no vertex in $N(u)$, except $v'$, adjacent to two vertices in $S'' \cup \{x\}$, then all vertices in $N(u) - \{v'\}$ are infected by $R$ at time $\geq 2$. Thus, $t(G,R,v') = 1$ and $t(G,R,u) \geq 3$. Let $q \in N(u)$ be the neighbor of $v'$ that has a neighbor $c \in S'$. Since  $S' \subseteq S'' \subseteq R$, then $c \in R$. Also, let $F = C_{R,v',1} \cup \{c\}$. We have that $t(G,R,q) \leq 2$ because $t(G,F,v') = 1$ and $t(G,F,c) = 0$ and, therefore, since $t(G,R,q) \leq 2$ and $t(G,F,v') = 1$, $t(G,F,u) \leq 3$. Since $t(G,F,u) \leq 3$, $|F| \leq 4$, $F \subseteq R$ and $t(G,R,u) \geq 3$, then $F$ and $R$ are the sets that we are looking for.
\end{proof}

\subsection{Proof of Claim \ref{claimbmf}}

\claimbmf*
\begin{proof}
First, let us start by defining the set $C_{H,v,t}$. For any vertex $v$ infected by some hull set $H$ at time $t$, let us define $C_{H,v,t}$ to be any minimal subset of $H$ such that the set $C_{H,v,t}$ infects $v$ at time $t$. In general, if $v$ is infected at time $t$ by $H$, then $|C_{H,v,t}| \leq 2^t$ because every vertex needs only two infected neighbors to be infected at any time $\geq 1$. Additionally, if $v$ is infected at time $t \geq 1$ and has a neighbor vertex infected at time 0 by $H$, since $v$ must have a neighbor $v'$ infected at time $t-1$, then, we have that $|C_{H,v,t}| \leq |C_{H,v',t-1}| + 1 \leq 2^{t-1} + 1$.\\

Let us begin the proof by proving the following fact:
\begin{fact}
There is a hull set $Z \supseteq S$ that infects $x$ at time 4 and some vertex $v \in N(x)$ at time 3 such that $|C_{Z,v,3}| \leq 4$.
\end{fact}

\begin{proof}
If there is some vertex $v \in N(x)$ infected at time 3 by $S$ such that $|C_{S,v,3}| \leq 4$, let $Z = S$ and we are done. So, assume that that is not the case.

Let us prove that there is a hull set $Z \supseteq S$ that infects $x$ at time 4 and some vertex $v \in N(x)$ at time 3 such that $v$ has a neighbor in $Z$ and another neighbor $s$ that is infected at time 2 and that has also a neighbor in $Z$. This should be enough because, if $s$ has a neighbor in $Z$, then $|C_{Z,s,2}| \leq 3$ and, if $v$ also has a neighbor in $Z$, then $|C_{Z,v,3}| \leq |C_{Z,s,2}| + 1 \leq 3 + 1 = 4$.

Since there is no vertex $v \in N(x)$ infected at time 3 by $S$ such that $|C_{S,v,3}| \leq 4$, then there is no vertex $v \in N(x)$ that is infected at time 3 by $S$ such that $v$ has a neighbor in $S$ and another neighbor $s$ that is infected at time 2 and has a neighbor in $S$.

So, if there is a vertex in $N(x)$ infected at time 3 by $S$ that has a neighbor in $S$, let $v$ be that vertex and let $P = S$. If not, let $v$ be any vertex in $N(x)$ that is infected at time 3 by $S$, has a neighbor $r$ infected at time 1 by $S$ and has another neighbor $s$ infected at time 2 by $S$ that has a neighbor in $S$. If there is no such vertex $v$, let $v$ be any vertex in $N(x)$ that is infected at time 3 by $S$ and has a neighbor $r$ infected at time 1 by $S$, or, if there is no such vertex $v$, let $v$ be any vertex in $N(x)$ that is infected at time 3 by $S$ and let $r$ be any vertex adjacent to $v$ infected at time 2 by $S$. Let $P = S \cup \{r\}$. 

We have that, since $G$ is bipartite, all vertices infected at time $\geq 2$ by $S$ that are adjacent to any vertex in $N(x)$ are at distance $\geq 2$ of $r$, except for $r$ itself, and, therefore, by Lemma \ref{lemadistanciak}, all vertices infected at time $\geq 2$ by $S$ that are adjacent to any vertex in $N(x)$ are infected at time $\geq 2$ by $P$. 

Additionally, if $r$ is infected at time 2 by $S$, by the choice of $r$, there is no vertex in $N(x)$ infected at time 3 by $S$ that has one neighbor infected at time $\leq 1$ by $S$ and, since all vertices infected at time $\geq 4$ by $S$ cannot be adjacent to $r$ and another vertex infected at time $\leq 1$ by $S$, then, if $r$ is infected at time 2 by $S$, all vertices in $N(x)$ that are infected at time $\geq 3$ by $S$ have at most one neighbor infected at time $\leq 1$ by $P$.

If $r$ is infected at time 1 by $S$, then all vertices in $N(x)$ infected at time $\geq 3$ by $S$ cannot be adjacent to $r$ and another vertex infected at time $\leq 1$ by $S$ because, otherwise, they would be infected at time $\leq 2$ by $S$. Then, if $r$ is infected at time 1 by $S$, all vertices in $N(x)$ that are infected at time $\geq 3$ by $S$ have at most one neighbor infected at time $\leq 1$ by $P$.

Thus, whichever the infection time of $r$ by $S$, we have that all vertices in $N(x)$ that are infected at time $\geq 3$ by $S$ have at most one neighbor infected at time $\leq 1$ by $P$. Thus, all vertices in $N(x)$ that are infected at time $\geq 3$ by $S$ are infected at time $\geq 3$ by $P$. Hence, $v$ is infected at time $\geq 3$ and $x$ at time $\geq 4$ by $P$. Thus, by Lemma \ref{corolariosubset}, $v$ is infected at time 3 and $x$ at time 4 by $P$.

Thus, there is a vertex $v$ in $N(x)$ that is infected at time 3 by $P$ such that $v$ has a neighbor that is infected at time 0 by the hull set $P$ and $x$ is infected at time 4 by $P$. At this point, if there is some vertex $v'$ in $N(x)$ such that $P$ infects $v'$ at time 3 and $|C_{P,v',3}| \leq 4$, let $Z = P$ and we are done. If not, henceforth, assume that there is not such vertex $v'$.

Note that if $v$ has a neighbor $r$ infected at time 1 by $S$ and has another neighbor infected at time 2 by $S$ that has a neighbor in $S$, then $|C_{P,v,3}| \leq 4$. Thus, since there is no vertex $v'$ in $N(x)$ such that $P$ infects $v'$ at time 3 and $|C_{P,v',3}| \leq 4$, and all vertices infected at time $\geq 2$ by $S$ and adjacent to some vertex in $N(x)$ are infected at time $\geq 2$ by $P$, then there is no vertex in $N(x)$ infected at time 3 by $P$ that has a neighbor infected at time 1 by $P$ and has another neighbor infected at time 2 by $P$ that has a neighbor in $P$.

Thus, we have that all vertices in $N(x)$ infected at time 3, including $v$, that are neighbor of some vertex in $P$ do not have a neighbor infected at time 2 that also has a neighbor in $P$. Then, let $s$ be any vertex adjacent to $v$ that is infected at time 2 by $P$. Since $v$ has a neighbor in $P$, we have that $s$ does not have a neighbor in $P$. Let $z$ be any vertex adjacent to $s$ that is infected at time 1 by $P$. Let $Z = P \cup \{z\}$. Let us prove that $Z$ infects $x$ at time 4, $v$ at time 3 and $s$ at time 2. Using the Lemma \ref{corolariosubset}, we have that is enough to prove that $Z$ infects $x$ at time $\geq 4$, $v$ at time $\geq 3$ and $s$ at time $\geq 2$. First, since $s$ has no neighbor in $P$, then it has at most one neighbor in $Z$. Therefore $s$ is infected at time $\geq 2$ by $Z$. Now, assume, by contradiction, that $Z$ infects some vertex $q \in N(x)$ that is infected at time $\geq 3$ by $P$ at time $\leq 2$. Since $q \neq z$ and $P$ infects $q$ at time 3 and $Z$ infects $q$ at time $\leq 2$, by Lemma \ref{lemadistanciak}, the vertices $z$ and $q$ are at distance either 1 or 2 of each other, but, since $G$ is bipartite, $z$ and $q$ are at distance 2 of each other. Therefore, since $q$ cannot have a neighbor infected at time 1 by $P$ and another neighbor $q'$ that is infected at time $\geq 2$ by $P$ that has a neighbor in $P$, we have that only two cases can happen.

In the first case, the vertex $q$ has a neighbor in $Z$ and another neighbor $q'$ that is infected at time $\geq 2$ by $P$ and at time 1 by $Z$. In this case, we have that $z$ is adjacent to $q'$ and, since $Z = P \cup \{z\}$, $q'$ must have a neighbor in $P$. Since $z$ and $q$ are at distance 2 of each other, then the vertex adjacent to $q$ in $Z$ cannot be $z$ and, therefore, $q$ has a neighbor in $P$. Thus, since $z$ is infected at time 1 by $P$ and $q$ and $q'$ have a neighbor in $P$, then, in fact, $q'$ is infected at time 2 and $q$ is infected at time 3 by $P$. Thus, since $q'$ has a neighbor in $P$, then $|C_{P,q',2}| \leq 3$. Also, since $q$ has a neighbor in $P$, we have that $|C_{P,q,3}| \leq |C_{P,q',2}| + 1 \leq 3 + 1 = 4$, which is a contradiction because we assumed that there is no vertex $v$ in $N(x)$ that is infected at time 3 by $P$ such that $|C_{P,v,3}| \leq 4$.

In the second, case we have that $q$ has two neighbors $q_1$ and $q_2$ that are infected at time $\geq 2$ by $P$ and at time 1 by $Z$. In this case, we have that $z$ is adjacent to both $q_1$ and $q_2$ and, also, we have that both $q_1$ and $q_2$ have a neighbor in $Z$ that is different from $z$, implying that both $q_1$ and $q_2$ have a neighbor in $P$. Thus, since $q_1$ and $q_2$ have a neighbor in $P$ and $z$ is infected at time 1 by $P$, then $q_1$ and $q_2$ are infected at time 2 by $P$ and, therefore, $q$ is infected at time 3 by $P$. So, since $z$ is adjacent to both $q_1$ and $q_2$, $z$ is infected by $P$ at time 1 and both $q_1$ and $q_2$, which are infected at time 2 by $P$, have a neighbor in $P$, then $|C_{P,q,3}| \leq |C_{P,q_1,2}| + |C_{P,q_2,2}| \leq |C_{P,z,1}| + 1 + 1 = 4$, which is a contradiction because we assumed that there is no vertex $v$ in $N(x)$ infected at time 3 by $P$ such that $|C_{P,v,3}| \leq 4$.

Thus, we have that $Z$ infects at time $\geq 3$ all vertices in $N(x)$ that are infected at time $\geq 3$ by $P$. Hence, $Z$ infects $v$ at time $\geq 3$ and $x$ at time $\geq 4$. Also, since $v$ and $s$ have a neighbor in $Z$, then $Z$ is the set we are looking for.
\end{proof}

Thus, there is a hull set $Z$ that is superset of $S$ that infects $x$ at time 4 and some vertex $v$ in $N(x)$ at time 3 such that $|C_{Z,v,3}| \leq 4$. Thus, if there is a vertex $v'$ in $N(x)$ that is infected at time $t \leq 2$ by $Z$, let $F = C_{Z,v,3} \cup C_{Z,v',t}$. Thus, we have that $|F| \leq |C_{Z,v,3}| + |C_{Z,v',t}| \leq 4 + 2^t \leq 8$. Since $F$ infects $v$ at time $\leq 3$ and $v'$ at time $\leq 2$, then $F$ infects $x$ at time $\leq 4$. Additionally, since $T_0 \cup N_{\geq 4}(u) \subseteq S \subseteq Z$ and $F \subseteq Z$, then $T_0 \cup N_{\geq 4}(u) \cup F \subseteq Z$. Since $x$ is infected by $Z$ at time 4 and by $F$ at time $\leq 4$, $|F| \leq 8$, and $T_0 \cup N_{\geq 4}(u) \cup F \subseteq Z$, letting $R = Z$, $R$ and $F$ are the sets we are looking for.

So, henceforth, assume that all vertices in $N(x)$ are infected at time $\geq 3$ by $Z$. Then, in this case, let us prove the following fact:

\begin{fact}
There is a hull set $R \supseteq Z$ and a vertex $v' \neq v$ in $N(x)$ that is infected at time 3 by $Z$ such that $R$ infects $x$ at time 4, infects $v$ at time $t \leq 3$, infects $v'$ at time 3, $|C_{R,v,t}| \leq 4$ and $|C_{R,v',3}| \leq 4$.
\end{fact}

\begin{proof}
To prove what we want, it is enough to prove that there is a hull set $R$ that infects $x$ at time 4 and some vertex $v' \in N(x)-\{v\}$ at time 3 such that $v'$ has a neighbor in $R$ and another neighbor $s$ that is infected by $R$ at time 2 and has a neighbor in $R$. This is because, since $v'$ and $s$ have a neighbor in $R$ then $|C_{R,s,2}| \leq 3$ and $|C_{R,v',3}| \leq |C_{R,s,2}| + 1 \leq 4$. Also, since $Z \subseteq R$ and $Z$ infects $v$ at time 3, then, by Lemma \ref{corolariosubset}, $R$ infects $v$ at time $t \leq 3$ and, since $v$ has a neighbor in $Z$, and consequently in $R$, then $|C_{R,v,t}| \leq 4$.

If there is a vertex $v' \in N(x)-\{v\}$ that is infected at time 3 by $Z$ such that $|C_{Z,v',3}| \leq 4$, let $R = Z$ and we are done. If not, henceforth, assume that there is not such vertex $v'$.

So, if there is a vertex in $N(x)-\{v\}$ infected at time 3 by $Z$ that has a neighbor in $Z$, let $v'$ be that vertex and let $P = Z$. If not, let $v'$ be any vertex in $N(x)-\{v\}$ that is infected at time 3 by $Z$, has a neighbor $r$ infected at time 1 by $Z$ and has another neighbor $s$ infected at time 2 by $Z$ that has a neighbor in $Z$. If there is no such vertex $v'$, let $v'$ be any vertex in $N(x)-\{v\}$ that is infected at time 3 by $Z$ and has a neighbor $r$ infected at time 1 by $Z$, or, if there is no such vertex $v'$, let $v'$ be any vertex in $N(x)-\{v\}$ that is infected at time 3 by $Z$ and let $r$ be any vertex adjacent to $v'$ infected at time 2 by $Z$. Let $P = Z \cup \{r\}$. 

We have that, since $G$ is bipartite, all vertices infected at time $\geq 2$ by $Z$ that are adjacent to any vertex in $N(x)$ are at distance $\geq 2$ of $r$, except for $r$ itself, and, therefore, by Lemma \ref{lemadistanciak}, all vertices infected at time $\geq 2$ by $Z$ that are adjacent to any vertex in $N(x)-\{v\}$ are infected at time $\geq 2$ by $P$. 

Additionally, if $r$ is infected at time 2 by $Z$, by the choice of $r$, there is no vertex in $N(x)-\{v\}$ infected at time 3 by $Z$ that has one neighbor infected at time $\leq 1$ by $Z$ and, since all vertices infected at time $\geq 4$ by $Z$ cannot be adjacent to $r$ and another vertex infected at time $\leq 1$ by $Z$, then, if $r$ is infected at time 2 by $Z$, all vertices in $N(x)-\{v\}$ that are infected at time $\geq 3$ by $Z$ have at most one neighbor infected at time $\leq 1$ by $P$.

If $r$ is infected at time 1 by $Z$, then all vertices in $N(x)-\{v\}$ infected at time $\geq 3$ by $Z$ cannot be adjacent to $r$ and another vertex infected at time $\leq 1$ by $Z$ because, otherwise, they would be infected at time $\leq 2$ by $Z$. Then, if $r$ is infected at time 1 by $Z$, all vertices in $N(x)-\{v\}$ that are infected at time $\geq 3$ by $Z$ have at most one neighbor infected at time $\leq 1$ by $P$.

Thus, whichever the infection time of $r$ by $Z$, we have that all vertices in $N(x)-\{v\}$ that are infected at time $\geq 3$ by $Z$ have at most one neighbor infected at time $\leq 1$ by $P$. Thus, all vertices in $N(x)-\{v\}$ that are infected at time $\geq 3$ by $Z$, which is all vertices in $N(x)-\{v\}$, are infected at time $\geq 3$ by $P$. Hence, $v'$ is infected at time $\geq 3$ and $x$ at time $\geq 4$ by $P$. Thus, by Lemma \ref{corolariosubset}, $v$ is infected at time 3 and $x$ at time 4 by $P$.

Thus, there is a vertex $v'$ in $N(x)-\{v\}$ that is infected at time 3 by $P$ such that $v'$ has a neighbor in $P$ and $x$ is infected at time 4 by $P$. At this point, if there is some vertex $w$ in $N(x)-\{v\}$ such that $P$ infects $w$ at time 3 and $|C_{P,w,3}| \leq 4$, let $R = P$ and we are done. If not, henceforth, assume that there is not such vertex $w$.

Note that if $v'$ has a neighbor $r$ infected at time 1 by $Z$ and has another neighbor infected at time 2 by $Z$ that has a neighbor in $Z$, then $|C_{P,v',3}| \leq 4$. Thus, since there is no vertex $w$ in $N(x)-\{v\}$ such that $P$ infects $w$ at time 3 and $|C_{P,w,3}| \leq 4$, and all vertices infected at time $\geq 2$ by $Z$ and adjacent to some vertex in $N(x)-\{v\}$ are infected at time $\geq 2$ by $P$, then there is no vertex in $N(x)-\{v\}$ infected at time 3 by $P$ that has a neighbor infected at time 1 by $P$ and has another neighbor infected at time 2 by $P$ that has a neighbor in $P$.

Thus, we have that all vertices in $N(x)-\{v\}$ infected at time 3, including $v'$, that are neighbor of some vertex in $P$ do not have a neighbor infected at time 2 that also has a neighbor in $P$. Then, let $s$ be any vertex adjacent to $v'$ that is infected at time 2 by $P$. Since $v'$ has a neighbor in $P$, we have that $s$ does not have a neighbor in $P$. Let $z$ be any vertex adjacent to $s$ that is infected at time 1 by $P$. Let $R = P \cup \{z\}$. Let us prove that $R$ infects $x$ at time 4, $v'$ at time 3 and $s$ at time 2. Using the Lemma \ref{corolariosubset}, we have that is enough to prove that $R$ infects $x$ at time $\geq 4$, $v'$ at time $\geq 3$ and $s$ at time $\geq 2$. First, since $s$ has no neighbor in $P$, then it has at most one neighbor in $R$. Therefore $s$ is infected at time $\geq 2$ by $R$. Now, assume, by contradiction, that $R$ infects some vertex $q \in N(x)-\{v\}$ that is infected at time $\geq 3$ by $P$ at time $\leq 2$. Since $q \neq z$ and $P$ infects $q$ at time 3 and $R$ infects $q$ at time $\leq 2$, by Lemma \ref{lemadistanciak}, the vertices $z$ and $q$ are at distance either 1 or 2 of each other, but, since $G$ is bipartite, $z$ and $q$ are at distance 2 of each other. Therefore, since $q$ cannot have a neighbor infected at time 1 by $P$ and another neighbor $q'$ that is infected at time $\geq 2$ by $P$ that has a neighbor in $P$, we have that only two cases can happen.

In the first case, the vertex $q$ has a neighbor in $R$ and another neighbor $q'$ that is infected at time $\geq 2$ by $P$ and at time 1 by $R$. In this case, we would have that $z$ is adjacent to $q'$ and, since $R = P \cup \{z\}$, $q'$ must have a neighbor in $P$. Since $z$ and $q$ are at distance 2 of each other, then the vertex adjacent to $q$ in $R$ cannot be $z$ and, therefore, $q$ has a neighbor in $P$. Thus, since $z$ is infected at time 1 by $P$ and $q$ and $q'$ have a neighbor in $P$, then, in fact, $q'$ is infected at time 2 and $q$ is infected at time 3 by $P$. Thus, since $q'$ has a neighbor in $P$, then $|C_{P,q',2}| \leq 3$. Also, since $q$ has a neighbor in $P$, we have that $|C_{P,q,3}| \leq |C_{P,q',2}| + 1 \leq 3 + 1 = 4$, which is a contradiction because we assumed that there is no vertex $w$ in $N(x)-\{v\}$ that is infected at time 3 by $P$ such that $|C_{P,w,3}| \leq 4$.

In the second case, we have that $q$ has two neighbors $q_1$ and $q_2$ that are infected at time $\geq 2$ by $P$ and at time 1 by $R$. In this case, we have that $z$ is adjacent to both $q_1$ and $q_2$ and, also, we have that both $q_1$ and $q_2$ have a neighbor in $R$ that is different from $z$, implying that both $q_1$ and $q_2$ have a neighbor in $P$. Thus, since $q_1$ and $q_2$ have a neighbor in $P$ and $z$ is infected at time 1 by $P$, then $q_1$ and $q_2$ are infected at time 2 by $P$ and, therefore, $q$ is infected at time 3 by $P$. So, since $z$ is adjacent to both $q_1$ and $q_2$, $z$ is infected by $P$ at time 1 and both $q_1$ and $q_2$, which are infected at time 2 by $P$, have a neighbor in $P$, then $|C_{P,q,3}| \leq |C_{P,q_1,2}| + |C_{P,q_2,2}| \leq |C_{P,z,1}| + 1 + 1 = 4$, which is a contradiction because we assumed that there is no vertex $w$ in $N(x)-\{v\}$ infected at time 3 by $P$ such that $|C_{P,w,3}| \leq 4$.

Thus, we have that $R$ infects at time $\geq 3$ all vertices in $N(x)-\{v\}$ that are infected at time $\geq 3$ by $P$, which is all vertices in $N(x)-\{v\}$. Therefore, $R$ infects $v'$ at time $\geq 3$ and $x$ at time $\geq 4$. Also, since $v'$ and $s$ have a neighbor in $R$, then $R$ is the set we are looking for.
\end{proof}

Thus, there is a hull set $R$ that is superset of $Z$ that infects $x$ at time 4 , some vertex $v \in N(x)$ at time $t \leq 3$ and some other vertex $v' \in N(x)$ at time 3 such that $|C_{R,v,t}| \leq 4$ and $|C_{R,v',3}| \leq 4$. Let $F = C_{R,v,t} \cup C_{R,v',3}$. Thus, we have that $|F| \leq |C_{R,v,t}| + |R_{Z,v',3}| \leq 4 + 4 = 8$. Since $F$ infects $v$ at time $\leq 3$ and $v'$ at time $\leq 3$, then $F$ infects $x$ at time $\leq 4$. Additionally, since $T_0 \cup N_{\geq 4}(u) \subseteq S \subseteq Z \subseteq R$ and $F \subseteq R$, then $T_0 \cup N_{\geq 4}(u) \cup F \subseteq R$. So, since $x$ is infected by $R$ at time 4 and by $F$ at time $\leq 4$, $|F| \leq 8$, and $T_0 \cup N_{\geq 4}(u) \cup F \subseteq R$, then $R$ and $F$ are the sets we are looking for.
\end{proof}


\section{Final Remarks}
In this paper, we showed the Percolation Time Problem is polynomial for a fixed $k = 3$ by finding polynomially computable characterization for graphs that have the percolation time $\geq 3$. Also, when the problem is restricted only to bipartite graphs, we showed that the Percolation Time Problem for a fixed $k \geq 5$ is NP-Complete. However, we found polynomially computable characterizations for graphs that have the percolation time $\geq 3$ and $\geq 4$. Regarding the polynomial cases, we derived from the characterizations algorithms that run in times $O(mn^5)$, $O(mn^3)$ and $O(m^2n^9)$ that solves the Percolation Time Problem when , respectively, the parameter $k$ is fixed in $3$ and, in bipartite graphs, when $k$ is fixed in $3$ and $4$.

The results in \cite{eurocomb13} along with the results in this work set the threshold for the fixed parameter $k$ in which the Percolation Time Problem ceases to be polynomial and becomes NP-complete.

We are hoping that these results, the results in \cite{eurocomb13} regarding NP-completeness of the Percolation Time Problem for planar graphs and future results regarding the NP-completeness of the Percolation Time Problem in restricted degree graphs shed some light in the following question: Can the Percolation Time Problem in subgraphs and induced subgraphs of $d$-dimensional grids be solved in polynomial time or this problem in these types of graphs is NP-Complete? 

Other interesting questions arises from this work, such as, can the complexity of the algorithms, which, in this case, is directly related to the maximum size of the sets of vertices that initially infects some vertex at time $k$, be improved? Also, is there a relation between the $P_3$-Carathéodory number~\cite{jayme2012} and the maximum size of the sets of vertices that initially infects some vertex at time $k$?


\end{document}